\documentclass[11pt]{article} 
\usepackage{amsfonts,amsmath,amssymb,amsthm} 
\usepackage{amsxtra}
\usepackage{complexity}
\usepackage{pstricks,pst-plot} 
\usepackage[normalem]{ulem}
\usepackage{array}
\usepackage[margin=2.66	cm]{geometry}
\usepackage{makecell}

\usepackage[active]{srcltx}
\usepackage{hyperref}  
\hypersetup{
    colorlinks,%
    citecolor=black,%
    filecolor=black,%
    linkcolor=black,%
    urlcolor=black
}

\renewenvironment{proof}{{\bfseries\noindent Proof.}}{\qed\vspace{3.5ex}}

\newtheorem*{lemma*}{Lemma}

\newtheorem*{theorem*}{Theorem}
\newtheorem{lemma}{Lemma}
\newtheorem{corollary*}{Corollary}
\newtheorem{corollary}{Corollary}
\newtheorem{theorem}{Theorem}
\newtheorem{definition}{Definition}

\newcommand{\tnc}[1]{}

\newcommand{\tne}[1]{\ensuremath{\langle #1\rangle}}
\newcommand{\PCPfour}{\mathcal{P}}
\newcommand{\CTS}{\mathcal{C}}
\newcommand{\CTSp}{\mathcal{C'}}
\newcommand{\TSTC}{\mathcal{T_{\CTS}}}
\newcommand{\TSTCp}{\mathcal{T'_{\CTS}}}
\def\Nset{\mathbb{N}}
\def\Zset{\mathbb{Z}}

\newcommand{\encodeOne}{\tne{1}}
\newcommand{\encodeOneSequence}{\ensuremath{b^{10}(ubb)^{\frac{x}{2}-7}u^{\frac{x}{2}+7}b^{2}ub^{x+2}}}
\newcommand{\encodeOneTrack}{\ensuremath{b^{10}(cbb)^{\frac{x}{2}-7}c^{\frac{x}{2}+7}b^{2}cb^{x+2}}}
\newcommand{\encodeOneTrackSpecialCaseA}{\ensuremath{b^9(cbb)^{\frac{x}{2}-7}c^{\frac{x}{2}+7}b^{2}cb^{x+2}}}

\newcommand{\encodeZero}{\tne{0}}
\newcommand{\encodeZeroSequence}{\ensuremath{b^4ub^2u^{x-1}b^2ub^{2x-8}}}
\newcommand{\encodeZeroTrack}{\ensuremath{b^4cb^2c^{x-1}b^2cb^{2x-8}}}
\newcommand{\encodeZeroTrackSpecialCaseA}{\ensuremath{b^3cb^2c^{x-1}b^2cb^{2x-8}}}

\newcommand{\encodeDeletion}{\tne{\epsilon}}
\newcommand{\encodeDeletionSequence}{\ensuremath{b^2ub^{3x-2}}}
\newcommand{\encodeDeletionTrack}{\ensuremath{b^2cb^{3x-2}}}
\newcommand{\encodeDeletionTrackCaseA}{\ensuremath{bcb^{3x-2}}}
\newcommand{\encodeDeletionTrackCaseB}{\ensuremath{b^2cb^{3x-3}}}
\newcommand{\encodeDeletionPrime}{\tne{\epsilon'}}
\newcommand{\encodeDeletionPrimeSequence}{\ensuremath{b^4ub^2u^{x-2}b^2ub^{2x-8}}}
\newcommand{\encodeDeletionPrimeTrack}{\ensuremath{b^4cb^2c^{x-2}b^2cb^{2x-8}}}

\newcommand{\track}[2]{\raisebox{2.5pt}{\ensuremath{\underset{[#1]}#2}}}

\begin{document}

\title{Undecidability in binary tag systems and the Post correspondence problem for four pairs of words}

\author{Turlough Neary\\
Institute of Neuroinformatics, University of Z\"urich and ETH Z\"urich, Switzerland\\
tneary@ini.phys.ethz.ch
}
\date{}
\maketitle

\begin{abstract}
Since Cocke and Minsky proved 2-tag systems universal, they have been extensively used to prove the universality of numerous computational models. Unfortunately, all known algorithms give universal 2-tag systems that have a large number of symbols. 
In this work, tag systems with only 2 symbols (the minimum possible) are proved universal via an intricate construction showing that they simulate cyclic tag systems. 
Our simulation algorithm has a polynomial time overhead, and thus shows that binary tag systems simulate Turing machines in polynomial time. 

We immediately find applications of our result. 
We reduce the halting problem for binary tag systems to the Post correspondence problem for 4 pairs of words.
This improves on 7 pairs, the previous bound for undecidability in this problem.
Following our result, only the case for 3 pairs of words remains open, as the problem is known to be decidable for 2 pairs. 
As a further application, we find that the matrix mortality problem is undecidable for sets with five $3\times 3$ matrices and for sets with two $15\times 15$ matrices. 
The previous bounds for the undecidability in this problem was seven $3\times 3$ matrices and two $ 21\times 21$ matrices.  
\end{abstract}

\section{Introduction}
Introduced by Post~\cite{Post1943}, tag systems have been used to prove Turing universality in numerous computational models, including some of the simplest known universal systems~\cite{Baiocchi2001,Cook2004,HarjuMargenstern2005,Kudlek2002A,Minsky1962,Hooper1969,Robinson1971,Robinson1991,Rogozhin1996,Rothemund1996,SiegelmannMargenstern1999}. Many universality results rely either on direct simulation of tag systems or on a chain of simulations the leads back to tag systems. Such relationships between models means that improvements in one model often has applications to many others. The results in~\cite{WoodsNeary2006B} are a case in point, where an exponential improvement in the time efficiency of tag systems had the domino effect of showing that many of the simplest known models of computation~\cite{Baiocchi2001,Cook2004,HarjuMargenstern2005,Hooper1969,Kudlek2002A,LindgrenNordahl1990,Minsky1962,Robinson1971,Robinson1991,Rogozhin1996,Rothemund1996,SiegelmannMargenstern1999} are in fact polynomial time simulators of Turing machines. Despite being central to the search for simple universal systems for 50 years, tag systems have not been the subject of simplification~since~the~early~sixties. 

In 1961, Minsky~\cite{Minsky1961} solved Post's longstanding open problem by showing that tag systems, with deletion number 6, are universal.
Soon after, Cocke and Minsky~\cite{Cocke1964} proved that tag systems with deletion number 2 (2-tag systems) are universal. Later, Hao Wang~\cite{Wang1963} showed that 2-tag systems with even shorter instructions were universal. 
The systems of both Wang, and Cocke and Minsky use large alphabets and so have a large number of rules.
Here we show that tag systems with only 2 symbols, and thus only 2 rules, are universal. Surprisingly, one of our two rules is trivial. We find immediate applications of our result. Using Cook's~\cite{Cook2004} reduction of tag systems to cyclic tag systems,  it is a straightforward matter to give a binary cyclic tag system program that is universal and contains only two 1 symbols. We also use our binary tag system construction to improve the bound for the number of pairs of words for which the Post correspondence problem~\cite{Post1946} is undecidable, and the bounds for the simplest sets of matrices for which the mortality problem~\cite{Paterson1970} is undecidable.

The search for the minimum number of word pairs for which the Post correspondence problem is undecidable began in the 1980s~\cite{Claus1980,Pansiot1981}. 
The best result until now was found by Matiyasevich and S\'{e}nizergues, whose impressive 3-rule semi-Thue system~\cite{Matiyasevich1996,Matiyasevich2005}, along with a reduction due to Claus~\cite{Claus1980}, showed that the problem is undecidable for 7 pairs of words. 
Improving on this undecidability bound of 7 pairs of words seemed like a challenging problem. In fact, Blondel and Tsitsiklis~\cite{Blondel2000} stated in their survey ``The decidability of the intermediate cases ($3\leqslant n\leqslant 6$) is unknown but is likely to be difficult to settle''. We give the first improvement on the bound of Matiyasevich and S\'{e}nizergues in 17 years: We reduce the halting problem for our binary tag system to the Post correspondence problem for 4 pairs of words. This leaves open only the case for 3 pairs of words, as the problem is known to be decidable for 2 pairs~\cite{Ehrenfeucht1982,Halava2002A}.

A number of authors~\cite{Blondel1997,Cassaigne1998,Halava2001,Halava2007,Paterson1970}, have used undecidability bounds for the Post correspondence problem to find simple matrix sets for which the mortality problem is undecidable. 
The matrix mortality problem is, given a set of $d\times d$ integer matrices, decide if the zero matrix can be expressed as a product of matrices from the set. Halava et al.~\cite{Halava2007} proved the mortality problem undecidable for sets with seven $3\times 3$ matrices, and using a reduction due Cassaigne and Karhum\"{a}ki~\cite{Cassaigne1998} they also showed the problem undecidable for sets with two $21\times 21$ matrices. 
Using our new bound, and applying the reductions used in~\cite{Cassaigne1998,Halava2001}, we find that the matrix mortality problem is undecidable for sets with five $3\times 3$ matrices and for sets with two $15\times 15$ matrices. In addition, by applying reductions due to Halava and Hirvensalo~\cite{Halava2007A}, we improve on previous undecidability bounds for a number of decision problems in sets that consist of two matrices. These new bounds include a set with two $7\times 7$ matrices for which the scalar reachability problem is undecidable.

We complete our introduction by recalling some decidability results and open problems for tag systems.  
Stephen Cook~\cite{Cook1966} proved that the reachability problem, and hence the halting problem, is decidable for non-deterministic
1-tag systems. More recently, De Mol~\cite{DeMol2010} has shown that the reachability (and thus halting) problem is decidable for binary 2-tag systems, a problem which Post~\cite{Post1965} claimed to have solved but never published. In the 1920s Post~\cite{Post1965} gave a simple binary 3-tag system ($0\rightarrow 00$, $1\rightarrow 1101$) whose halting problem is still open~\cite{DeMol2011}. De Mol~\cite{DeMol2008} reduced the well know Collatz problem to the halting problem for a remarkably simple 2-tag system that has 3 rules. The simple tag systems of Post and De Mol suggest that improving on existing decidability results would be quite difficult.

\section{Preliminaries}
We write $c_1\vdash c_2$ if a configuration $c_2$ is obtained from $c_1$ via a single computation step. We let $c_1\vdash^{t}c_2$ denote a sequence of $t$ computation steps. The length of a word $w$ is denoted $|w|$, and $\epsilon$ denotes the empty word. We let $\langle v\rangle$ denote the encoding of $v$, where $v$ is a symbol or a word. We use the standard binary modulo operation $a=m\bmod n$, where $a=m-ny$, $0\leqslant a<n$, and $a,m,n,$ and $y$ are integers.

\subsection{Tag systems}\label{sect:Tag systems}
\begin{definition}\label{def:Tag System}
A tag system consists of a finite alphabet of symbols $\Sigma$, a finite set of rules $R : \Sigma\rightarrow\Sigma^{\ast}$ and a deletion number $\beta\in\Nset$, $\beta\geqslant 1$.
\end{definition}
The tag systems we consider are deterministic. 
The computation of a tag system acts on a word $w=w_0 w_1\ldots w_{|w|-1}$ (here $w_i\in\Sigma$) which we call the \emph{dataword}. The entire configuration is given by $w$. 
In a computation step, the symbols~$w_0 w_1\ldots w_{\beta-1}$ are deleted and we apply the rule for $w_0$, i.e. a rule of the form $w_0\rightarrow w_{0,1}w_{0,2}\ldots w_{0,e}$, by appending the word $w_{0,1}w_{0,2}\ldots w_{0,e}$ (here $w_{0,j}\in\Sigma$). A dataword (configuration) $w'$ is obtained from $w$ via a single computation step as follows:
\begin{equation*}
w_0 w_1\ldots w_{\beta}\ldots w_{|w|-1} \vdash w_{\beta}\ldots w_{|w|-1}w_{0,1}w_{0,2}\ldots w_{0,e}
\end{equation*}
where $w_0\rightarrow w_{0,1}w_{0,2}\ldots w_{0,e} \in R$. 
A tag system halts if~$|w|< \beta $. As an example we give the first 5 steps of Post's~\cite{Post1965} binary tag system with deletion number 3 and the rules $0\rightarrow 00$ and $1\rightarrow 1101$ on the input $0101110$.
\begin{equation*}
010 1110 \quad\vdash\quad 111 000 \quad\vdash\quad 000 1101 \quad\vdash\quad 110 100  \quad\vdash\quad 100 1101 \quad\vdash\quad 11011101 \quad\vdash\quad \cdots  
\end{equation*}
We use the term \emph{round} to describe the $\lfloor \frac{|w|}{\beta}\rfloor$ or $\lceil \frac{|w|}{\beta}\rceil$ computation steps that traverse the word $w$ exactly once.
We say a symbol $w_0$ is \emph{read} if and only if at the start of a computation step it is the leftmost symbol (i.e. the rule $w_0\rightarrow w_{0,0}w_{0,1}\ldots w_{0,c}$ is applied), and we say a word $w=w_0 w_1\ldots w_{|w|-1}$ is \emph{entered with shift $z<\beta$} if $w_z$ is the leftmost symbol that is read in $w$. 
For example, in Figure~\ref{fig:Example shift change} the words $q$, $r$, and $v$ are entered with shifts of 0, 3, and 2 respectively.
We let $\track{z}{w}$ denote the sequence of symbols that is read during a single round on $w$ when it is entered with shift $z$, and we call $\track{z}{w}$ a track of $w$. 
If $w=w_0w_1\ldots w_{|w|-1}$, then $\track{z}{w}=w_{z}w_{z+\beta}w_{z+2\beta}w_{z+3\beta},\ldots,w_{z+l\beta}$ where $|w|-\beta\leqslant z+l\beta< |w|$. 
For example, in Figure~\ref{fig:Example shift change} we have $\track{0}{w}=q_0q_4r_3r_7v_2$. 
A word $w$ has a \emph{shift change} of $0\leqslant s<\beta$ if $|w|=y\beta-s$ where $y>0$ is a natural number. 
 
\begin{figure}
\setlength{\unitlength}{.6cm}
\begin{picture}(10,2.1)
\put(.8,1){$q_0q_1q_2q_3q_4\;r_0r_1r_2r_3r_4r_5r_6r_7r_8\;v_0v_1v_2v_3v_4v_5
\qquad\vdash\qquad 
q_4\;r_0r_1r_2r_3r_4r_5r_6r_7r_8\;v_0v_1v_2v_3v_4v_5
\;\;\quad\vdash\qquad$} 
\put(.8,0){$r_3r_4r_5r_6r_7r_8\;v_0v_1v_2v_3v_4v_5
\qquad \vdash\qquad 
r_7r_8\;v_0v_1v_2v_3v_4v_5
\qquad \vdash\qquad 
v_2v_3v_4v_5$}
\end{picture}
\caption{Four computation steps of a tag system with deletion number $\beta=4$ on the word $w=qrv$. Here $q=q_0q_1q_2q_3q_4$, $r=r_0r_1r_2r_3r_4r_5r_6r_7r_8$, $v=v_0v_1v_2v_3v_4v_5$, and $q_i$, $r_i$ and $v_i$ are tag system symbols, and for simplicity we assume that all symbols append the empty word (i.e.~all rules have the form $w_i\rightarrow \epsilon$).}\label{fig:Example shift change}
\end{figure}

\begin{lemma}\label{lem:shiftChange}
Given a tag system $T$ with deletion number $\beta$ and the word $rv\in\Sigma^{\ast}$, where the word $r$ has a shift change of $s$ and $|v|\geqslant\beta$, after one round of $T$ on $r$ entered with shift $z$ the word $v$ is entered with shift $(z+s)\bmod\beta$.
\end{lemma}
Before we give the proof of Lemma~\ref{lem:shiftChange} we note that Figure~\ref{fig:Example shift change} gives examples of the shift change caused by reading a word: In Figure~\ref{fig:Example shift change} the word $q$ is entered with shift $0$ and has shift change of $s=3$ and so $r$ is entered with shift $3=(0+3)\bmod 4$, and the word $r$ has a shift change of $3$ so the word $v$ is entered with shift $2=(3+3)\bmod 4$.\\

\begin{proof}
Recall that when $r$ is entered with shift $z$ then $r_z$ is the leftmost symbol read in $r$. In Equation~\eqref{eq: shift change} the sequence of symbols read in $r$ and the leftmost symbol read in $v$ are given in bold. The rightmost symbol read in $r$ is $r_{z+l\beta}$ for some $l\in\Nset$ and the next symbol read is $v_{q}$ in $v$. It follows that symbols $r_{z+l\beta}$ to $v_{q-1}$ are deleted in the computation step when $r_{l\beta+z}$ is read. Since $\beta$ symbols are deleted at each computation step, from Equation~\eqref{eq: shift change} we get $y\beta-s-l\beta-z+q=\beta$ which we rewrite as $q=\beta(l-y+1)+s+z$. The shift value $q$ must be $<\beta$ which gives $q=(s+z)\bmod\beta$.
\begin{equation}\label{eq: shift change}
r_0\dots r_{z-1}\pmb{r_z}r_{z+1}\ldots r_{z+\beta-1}\pmb{r_{z+\beta}}r_{z+\beta+1}\,\ldots\, r_{z+l\beta-1}\pmb{r_{z+l\beta}}r_{z+l\beta+1}\,\ldots\,r_{y\beta-s-1}v_{0}\,\ldots\, v_{q-1}\pmb{v_{q}} v_{q+1}\ldots
\end{equation}
\end{proof}

\begin{lemma}\label{lem:read upper or lower bound number of symbols}
Given a tag system $T$ with deletion number $\beta$ and the word $w$, where $|w|=y\beta-s$ with $y\in\Nset$ and $0<s<\beta$, one round of $T$ on $w$ entered with shift $z<\beta$ reads $\lceil\frac{|w|}{\beta}\rceil=y$ symbols if $z<\beta-s$, and $\lfloor\frac{|w|}{\beta}\rfloor=y-1$ symbols if $z\geqslant\beta-s$.
\end{lemma}
\begin{proof}
In Equation~\eqref{eq: enter w with shift beta-s-1} the sequence of symbols read when $w$ is entered with shift $\beta-s-1$ is given in bold. This bold sequence has length $y$.
It is fairly straightforward to see that if we enter $w$ with a shift $<\beta-s$ we read $y$ symbols and if we enter $w$ with a shift $\geqslant \beta-s$ we read $y-1$ symbols.
\begin{equation}\label{eq: enter w with shift beta-s-1}
w_0\,w_1\;\ldots\; w_{\beta-s-2}\,\pmb{w_{\beta-s-1}}\,w_{\beta-s}\;\ldots\; w_{2\beta-s-2}\,\pmb{w_{2\beta-s-1}}\,w_{2\beta+s}\;\ldots\; w_{y\beta-s-2}\,\pmb{w_{y\beta-s-1}}
\end{equation}
\end{proof}

\subsection{Cyclic tag systems}\label{sec:Cyclic tag systems}
\begin{definition}\label{def:Cyclic tag system}
A cyclic tag system $\CTS=\alpha_0,\ldots,\alpha_{p-1},$ is a list 
of words $\alpha\in\{0,1\}^{\ast}$ called appendants.
\end{definition}
A {\em configuration} of a cyclic tag system consists of (i) a {\em marker} that points to a single appendant~$\alpha_{m}$ in~$\CTS$, and (ii) a word $w=w_1\ldots w_{|w|}\in\{0,1\}^*$. We call~$w$ the {\em dataword}. Intuitively the list $\CTS$ is a program with the marker pointing to instruction~$\alpha_{m}$. In the initial configuration the marker points to appendant~$\alpha_0$ and~$w$ is the binary input word.
\begin{definition}\label{def:Computation CTS}
A computation step is deterministic and acts on a configuration in one of 
two ways:
\begin{itemize}
\item If $w_1 = 0$ then $w_1$ is deleted and the marker moves to appendant 
$\alpha_{(m+1\bmod p)}$.
\item If $w_1 = 1$ then $w_1$ is deleted, the word $\alpha_{m}$ is 
appended onto the right end of~$w$,  and the marker moves to appendant 
$\alpha_{(m+1\bmod p)}$.
\end{itemize}
\end{definition}
A cyclic tag system completes its computation if (i) the dataword is the empty word or (ii) it enters a repeating sequence of configurations.

As an example we give first 6 steps of the cyclic tag system $\CTS=001,01,11$ on the input word $101$. In each configuration $\CTS$ is given on the left with the marked  appendant highlighted in bold~font.
\begin{eqnarray*}
\begin{split}
& \pmb{001},01,11\quad 101  & \;\vdash\;\quad  & 001,\pmb{01},11\quad 01001 &\;\vdash\;\quad  &001,01,\pmb{11}\quad 1001 & \;\vdash\;\quad  & \pmb{001},01,11\quad 00111\\ 
 \;\vdash\;\quad  & 001,\pmb{01},11\quad 0111 & \;\vdash\;\quad  & 001,01,\pmb{11}\quad 111 & \;\vdash\;\quad  & \pmb{001},01,11\quad 1111 & \;\vdash\;\quad &\cdots
\end{split}
\end{eqnarray*}
Cyclic tag systems were introduced by Cook~\cite{Cook2004} and used to prove the cellular automaton Rule~110 universal. We gave an exponential improvement in the time efficiency of cyclic tag~systems~to~show:
\begin{theorem}[\cite{NearyWoods2006C}]\label{thm:CTS simulate TMs}
Let $M$ be a single-tape deterministic Turing machine that computes in time~$t$. Then there is a cyclic tag system $\CTS$ that simulates the computation of M in time $O(t^3 \log t)$. 
\end{theorem}
Given a cyclic tag system $\CTS=\alpha_0,\ldots,\alpha_{p-1}$ we can construct another cyclic tag system $\CTS'$ by concatenating an arbitrary number of copies of the program for $\CTS$. The system $\CTS'$ simulates step for step the computation of $\CTS$. For example take $\CTS=001,01,11$ given above, if we define $\CTS'=001,01,11,001,01,11$ then $\CTS'$ will give the same sequence of computation steps for any computation of $\CTS$.

\section{Simulating cyclic tag systems with binary tag systems}\label{sec:simulating CTS with TS} 

In Theorem~\ref{thm:main theorem}, the tag system $\TSTC$ simulates an arbitrary cyclic tag system with a program of length $3k+2$ where $k\in\Nset$. 
We do not lose generality with this restriction since the cyclic tag systems given by the construction in~\cite{NearyWoods2006C} satisfy this condition and are Turing universal.

\begin{theorem}\label{thm:main theorem}
Let $\CTS=\alpha_0,\alpha_1,\ldots\alpha_{3k+1}$ with $k\in\Nset$ be a cyclic tag system that runs in time~$t$. Then there is a binary tag system $\TSTC$ that simulates the computation of $\CTS$ in time $O(t^2)$.
\end{theorem}

\subsubsection{Cyclic tag system $\CTSp$ and binary tag system $\TSTC$}\label{sec:Cyclic tag system C'}
Given the program $\CTS=\alpha_0,\alpha_1,\ldots\alpha_{3k+1}$, we can give a cyclic tag system $\CTSp=(\alpha_0,\alpha_1,\ldots\alpha_{3k+1})^{q}$ of length $q(3k+2)$ that simulates $\CTS$ step for step when given the same input dataword as $\CTS$ (see the last paragraph of Section~\ref{sec:Cyclic tag systems}). 
The value $q\in\Nset$ is chosen so that $q(3k+2)=3x-2$ for $x\in\Nset$, such that $0=x\bmod2$ and $r<\frac{x}{2}-7$, where $r$ is the length of the longest appendant in $\CTS$. 

We construct a binary tag system $\TSTC$ that simulates the computation of $\CTSp$. The deletion number of $\TSTC$ is $\beta$, its alphabet is $\{b,c\}$, and its rules are of the form $b\rightarrow b$ and $c\rightarrow u$, where $u\in\{b,c\}^\ast$. The binary word $u$ encodes the entire program of $\CTSp$ and is defined by Tables~\ref{tab:Tracks in u for Case 1} to~\ref{tab:Tracks in u for Case 2b}. We will explain how to read these tables later in Section~\ref{sec:simulating CTS with TS}.

\subsubsection{Encoding used by $\TSTC$}\label{sec:Encoding used by tag system}
The cyclic tag system symbols 0 and 1 are encoded as the binary words $\encodeZero=\encodeZeroSequence$ and $\encodeOne=\encodeOneSequence$ respectively. We refer to $\encodeZero$ and $\encodeOne$ as objects.

\begin{definition}\label{def:input encoding}
An arbitrary input dataword $w_1w_2\ldots w_n\in\{0,1\}^\ast$ to a cyclic tag system is encoded as the $\TSTC$ input dataword $\tne{w_1}\tne{w_2}\ldots \tne{w_n}$.
\end{definition}
During the simulation we make use of two extra objects: the binary words $\encodeDeletion=\encodeDeletionSequence$ and $\encodeDeletionPrime=\encodeDeletionPrimeSequence$. An arbitrary (not necessarily input) cyclic tag system dataword $w_1w_2\ldots w_l\in\{0,1\}^\ast$ is encoded~as
\begin{equation}\label{eq:encoding CTS configuration}
\tne{w_1,z}\{\encodeDeletion,\encodeDeletionPrime\}^\ast\tne{w_2}\{\encodeDeletion,\encodeDeletionPrime\}^\ast\tne{w_3}\ldots \{\encodeDeletion,\encodeDeletionPrime\}^\ast\tne{w_l}\{\encodeDeletion,\encodeDeletionPrime\}^\ast 
\end{equation}
where $\tne{w_1,z}$ denotes the word given by an object $\tne{w_1}\in\{ \encodeZero , \encodeOne \}$ with its leftmost $z<\beta$ symbols deleted. This implies that $\tne{w_1}$ is entered with the shift value $z$ from Table~\ref{tab:equalities}. Finally, each appendant $\alpha_m=\sigma_1\sigma_2\ldots\sigma_v$ of $\CTSp$ is encoded via Equations~\eqref{eq:alpha m},~\eqref{eq:alphaPrime m1} or~\eqref{eq:alphaPrime m2}, (where $\sigma_i\in\{0,1\}$).\\
\begin{minipage}{.4\textwidth}
 \begin{equation}
  \tne{\alpha_m}=\tne{\sigma_1}\tne{\sigma_2}\ldots\tne{\sigma_v}\encodeDeletion^{x-v+1}\;\label{eq:alpha m}
 \end{equation}
\end{minipage}
\begin{minipage}{.59\textwidth}
\begin{align}
&\;\;\;\tne{\alpha'_m}= \tne{\sigma_1}\tne{\sigma_2}\ldots\tne{\sigma_j}\encodeDeletionPrime\tne{\sigma_{j+1}}\ldots\tne{\sigma_v}\encodeDeletion^{x-v} \label{eq:alphaPrime m1} \\
&\;\;\;\tne{\alpha'_m}=\tne{\sigma_1}\tne{\sigma_2}\ldots\tne{\sigma_v}\encodeDeletion^{j-v}\encodeDeletionPrime\encodeDeletion^{x-j} \label{eq:alphaPrime m2}
\end{align}
\end{minipage}

\begin{table}
\centering
\begin{tabular}{|@{\:}l@{}|}
\hline
$|\encodeDeletion|=(3x+1)\beta$,\quad\;\; $|\encodeDeletion|=|u|+3x$,\quad\;\; $|u|=(3x+1)\beta-3x$,\quad\;\,  $|\encodeOne|=|\encodeZero|=(x+1)|u|+2x$,
\\  
$|\encodeOne|=|\encodeZero|=(x+1)((3x+1)\beta-3x)+2x$,\qquad\;
$z_1=3x^2+x$,\qquad\; $z_1(3x-2)=\beta$,\quad\\ 
$|\encodeDeletionPrime|=x|u|+2x$,\qquad\;  
$|\encodeDeletionPrime|=x((3x+1)\beta-3x)+2x$,\qquad\; $z_2=3x^2-2x$,\\
$z=((z_1m+z_2d)\bmod\beta)$,\qquad\; $0\leqslant m< 3x-2$,\qquad\; $0\leqslant d< 3x+1$\\
\hline
\end{tabular}
\caption{Length of objects and shift change values. The shift change for $\encodeOne$ and $\encodeZero$ is $z_1$, the shift change for $\encodeDeletionPrime$ is $z_2$, the deletion number of $\TSTC$ is $\beta$, and $3x-2$ is the number of appendants in $\CTS'$. The value of $x$ is given in Section~\ref{sec:Cyclic tag system C'}.}\label{tab:equalities}
\end{table}

\subsubsection{Lengths of objects and shift values}\label{sec:length of objects}
The sequence of symbols that is read in a word is determined by the shift value with which it is entered (see for example Figure~\ref{fig:u track and encodedOne track} (i)).
So in the simulation we use the shift value for algorithm control flow. 
In Table~\ref{tab:equalities} we give the length of objects $\encodeZero$, $\encodeOne$, $\encodeDeletion$, and $\encodeDeletionPrime$, and their shift change values.
Recall from Section~\ref{sect:Tag systems} that an object of length $y\beta-s$ has a shift change of $s$, where $s<\beta$ and $\beta$ is the deletion number. 
So, from the object lengths $|\encodeOne|=|\encodeZero|$ and $|\encodeDeletionPrime|$ we get the respective shift change values of $z_1$ and $z_2$ in Table~\ref{tab:equalities}.
From Lemma~\ref{lem:shiftChange}, the shift an object is entered with is determined by the shift change of the objects previously read in the dataword. 
So when we have a dataword containing only $\encodeOne$, $\encodeZero$, $\encodeDeletion$ and $\encodeDeletionPrime$ objects, we enter objects with shifts of the form $z=((z_1m+z_2d)\bmod\beta)$ (see Table~\ref{tab:equalities}). 
The range of values for $m$ and $d$ in Table~\ref{tab:equalities} covers all possible shift values. 
To see this note that when $m=3x-2$, then  $z_1m=\beta$ giving $0=z_1m\bmod\beta$, and when $d=3x+1$, then $z_2d=\beta$ giving $0=z_2d\bmod\beta$.

\begin{figure}
\setlength{\unitlength}{.6cm}
\begin{picture}(10,4.1)
\put(0,3){(i)}
\put(1.8,3){$\tne{1,z}\tne{a_2}\ldots\tne{a_h}\qquad\vdash^{\lceil\frac{|\encodeOne|}{\beta}\rceil}\quad\;\; \tne{a_2,z+z_1}\tne{a_3}\ldots\tne{a_h}\tne{\alpha_m}$} 
\put(0,2){(ii)}
\put(1.8,2){$\tne{0,z}\tne{a_2}\ldots\tne{a_h}\qquad\vdash^{\lceil\frac{|\encodeZero|}{\beta}\rceil}\quad\;\; \tne{a_2,z+z_1}\tne{a_3}\ldots\tne{a_h}\encodeDeletion^{x+1}$} 
\put(0,1){(iii)}
\put(1.8,1){$\tne{\epsilon',z}\tne{a_2}\ldots\tne{a_h}\quad\;\;\;\vdash^{\lceil\frac{|\encodeDeletionPrime|}{\beta}\rceil}\;\;\;\;\; \tne{a_2,z+z_2}\tne{a_3}\ldots\tne{a_h}\encodeDeletion^{x}$} 
\put(0,0){(iv)}
\put(1.8,0){$\tne{\epsilon,z}\tne{a_2}\ldots\tne{a_h}\qquad\,\vdash^{\frac{|\encodeDeletion|}{\beta}}\quad\;\;\;\;\: \tne{a_2,z}\tne{a_3}\ldots\tne{a_h}\encodeDeletion$}
 \end{picture}
\caption{Objects $\encodeOne$, $\encodeZero$, $\encodeDeletion$ and $\encodeDeletionPrime$ being read by $\TSTC$ when entered with shift $z$, where $a_i\in\{\encodeDeletion,\encodeDeletionPrime,\encodeZero,\encodeOne\}$, and $z=(z_1m+z_2d)\bmod\beta$. In (i) and (ii) $z<\beta-z_1$, in (iii) $z<\beta-z_2$, and in (iv) $z<\beta$. The encoded appendant $\tne{\alpha_{m}}$ is given in Equation~\eqref{eq:alpha m}, and the values $z_1$, $z_2$, $m$, and $d$ are given in Table~\ref{tab:equalities}.}\label{fig:reading objects when number of symbols read is upperbound}
\end{figure}

\subsection{The simulation algorithm}\label{sec:Simulation algorithm}
Here we give a high level picture of our algorithm using Figures~\ref{fig:reading objects when number of symbols read is upperbound} and~\ref{fig:reading objects when number of symbols read is lowerbound}. Following this, in Sections~\ref{sec:Encoding the marked appendant} and~\ref{sec:Reading objects and $u$ subwords}, the lower level details are then given. 

\subsubsection{Algorithm overview}\label{sec:Algorithm overview}
Figures~\ref{fig:reading objects when number of symbols read is upperbound} and~\ref{fig:reading objects when number of symbols read is lowerbound} give arbitrary examples that cover all possible cases for reading each of the four objects. 
In both figures $\vdash^y$ denotes the $y$ computation steps that read the entire leftmost object in the dataword on the left and produce the new dataword on the right. 
For example, in Figure~\ref{fig:reading objects when number of symbols read is upperbound} (ii) when $\encodeZero$ is read it appends $\encodeDeletion^{x+1}$ in $\lceil\frac{|\encodeZero|}{\beta}\rceil$ computation steps.
There are two cases for reading $\encodeZero$, $\encodeOne$, and $\encodeDeletionPrime$ objects with each case determined by the number of symbols read in the object (see Lemma~\ref{lem:read upper or lower bound number of symbols}). 
There is only one case for reading $\encodeDeletion$ as $\lceil\frac{|\encodeDeletion|}{\beta}\rceil=\lfloor\frac{|\encodeDeletion|}{\beta}\rfloor$. 
The objects $\encodeDeletion$ and $\encodeDeletionPrime$ are garbage objects that have no effect on the simulation.
To see this note from Figures~\ref{fig:reading objects when number of symbols read is upperbound} and~\ref{fig:reading objects when number of symbols read is lowerbound} that $\encodeDeletion$ and $\encodeDeletionPrime$ objects append only more garbage objects, and as we will see in Section~\ref{sec:Encoding the marked appendant} the shift change caused by reading an $\encodeDeletion$ or an $\encodeDeletionPrime$ does not effect algorithm control flow.
The garbage objects are introduced to simulate deletion as our binary tag system has no rule that appends the empty word~$\epsilon$.

When $\TSTC$ reads an $\encodeOne$ or an $\encodeZero$ object as shown in (i) and (ii) of Figures~\ref{fig:reading objects when number of symbols read is upperbound} and~\ref{fig:reading objects when number of symbols read is lowerbound} it simulates a computation step where $\CTSp$ reads a $1$ or a $0$.
At the beginning of the simulated computation step the currently marked appendant $\alpha_m$ is encoded by the shift value $z=(z_1m+z_2d)\bmod\beta$. 

\emph{Simulating the Definition~\ref{def:Computation CTS} computation step read a 0:} In (ii) of Figures~\ref{fig:reading objects when number of symbols read is upperbound} and~\ref{fig:reading objects when number of symbols read is lowerbound} we have the two possible cases for reading a $0$. In both cases when an $\encodeZero$ is read it gets deleted and only garbage objects are appended simulating that $\CTSp$ appends nothing. After reading $\encodeZero$ the adjacent object $\tne{a_2}$ is entered with shift $(z_1(m+1)+z_2d)\bmod\beta$, simulating that the next appendant $\alpha_{((m+1)\bmod (3x-2))}$ is marked.

\emph{Simulating the Definition~\ref{def:Computation CTS} computation step read a 1:}
In (i) of Figures~\ref{fig:reading objects when number of symbols read is upperbound} and~\ref{fig:reading objects when number of symbols read is lowerbound} we have the two possible cases for reading a $1$. In both cases when an $\encodeOne$ is entered with shift $z=((z_1m+z_2d)\bmod\beta)$ it is deleted and the encoding of $\alpha_m$ is appended. After reading $\encodeOne$ the adjacent object $\tne{a_2}$ is entered with shift $(z_1(m+1)+z_2d)\bmod\beta$, simulating that the next appendant $\alpha_{((m+1)\bmod (3x-2))}$ is marked.

\begin{figure}
\setlength{\unitlength}{.6cm}
\begin{picture}(10,3.2)
\put(0,2){(i)}
\put(1.8,2){$\tne{1,z}\tne{a_2}\ldots\tne{a_h}\qquad\vdash^{\lfloor\frac{|\encodeOne|}{\beta}\rfloor}\quad \tne{a_2,(z+z_1)\bmod\beta}\,\tne{a_3}\ldots\tne{a_h}\tne{\alpha'_{m}}$}
\put(0,1){(ii)}
\put(1.8,1){$\tne{0,z}\tne{a_2}\ldots\tne{a_h}\qquad\vdash^{\lfloor\frac{|\encodeZero|}{\beta}\rfloor}\quad \tne{a_2,(z+z_1)\bmod\beta}\,\tne{a_3}\ldots\tne{a_h}\encodeDeletion^{j}\encodeDeletionPrime\encodeDeletion^{x-j}$}
\put(0,0){(iii)}
\put(1.8,0){$\tne{\epsilon',z}\tne{a_2}\ldots\tne{a_h}\quad\;\;\;\vdash^{\lfloor\frac{|\encodeDeletionPrime|}{\beta}\rfloor}\;\;\; \tne{a_2,(z+z_2)\bmod\beta}\,\tne{a_3}\ldots\tne{a_h}\encodeDeletion^{j}\encodeDeletionPrime\encodeDeletion^{x-j-1}$} 
\end{picture}
\caption{Objects $\encodeZero$, $\encodeOne$ and $\encodeDeletionPrime$ being read by $\TSTC$ when entered with shift $z$, where $a_i\in\{\encodeDeletion,\encodeDeletionPrime,\encodeZero,\encodeOne\}$, and  $z=(z_1m+z_2d)\bmod\beta$. In (i) and (ii) $z\geqslant\beta-z_1$, and in (iii) $z\geqslant\beta-z_2$. The encoded appendant $\tne{\alpha'_{m}}$ is given in Equations~\eqref{eq:alphaPrime m1} and~\eqref{eq:alphaPrime m2}, and the values $z_1$, $z_2$, $m$, and $d$ are given in Table~\ref{tab:equalities}.}\label{fig:reading objects when number of symbols read is lowerbound}
\end{figure}

In both cases above  there is one further step that is needed to complete the simulation of the Definition~\ref{def:Computation CTS} computation step. 
Note from Equation~\eqref{eq:encoding CTS configuration} that between a pair of encoded cyclic tag system symbols $\tne{w_1}$ and $\tne{w_{2}}$ is a word of the from $\{\encodeDeletion,\encodeDeletionPrime\}^\ast$. 
So after $\tne{w_1}$ is read, the word $\{\encodeDeletion,\encodeDeletionPrime\}^\ast$ is read placing the object $\tne{w_{2}}$ at the left end of the dataword. This completes the simulated computation step as $\TSTC$ is now ready to begin reading the next encoded symbol $\tne{w_{2}}$. 
At the end of the next section we see that reading the garbage objects $\encodeDeletion$ and $\encodeDeletionPrime$ does not change the appendant encoded in the shift which means that $\tne{w_{2}}$ is entered with a shift value encoding the correct appendant $\alpha_{((m+1)\bmod (3x-2))}$.

We have not yet described how reading the objects $\encodeOne$, $\encodeZero$, $\encodeDeletion$, and $\encodeDeletionPrime$ append the appendants shown in Figures~\ref{fig:reading objects when number of symbols read is upperbound} and~\ref{fig:reading objects when number of symbols read is lowerbound}. To do so we must give the sequence of symbols read in each object when entered with shift $z$. 
The word $u$ that appears in the objects $\encodeOne$, $\encodeZero$, $\encodeDeletion$, and $\encodeDeletionPrime$ is defined such that the $u$ words read in these objects append the appendants shown in Figures~\ref{fig:reading objects when number of symbols read is upperbound} and~\ref{fig:reading objects when number of symbols read is lowerbound}. 
The $b$ symbols that appear in each object are used to control the shift with which we enter each $u$ within an object and thus control the sequence of symbols read in each $u$ (see Figure~\ref{fig:u track and encodedOne track} (i)). 
For example, if we enter an $\encodeDeletion=\encodeDeletionSequence$ with shift $z$ then the leftmost pair of $b$ symbols cause the $u$ to be entered with shift $z-2$ and we read track $\track{z-2}{u}$ (in this case $s=z-2$ in Figure~\ref{fig:u track and encodedOne track} (i)). 
If we assign track $\track{z-2}{u}$ a value that will append an $\encodeDeletion$, then when $\encodeDeletion$ is entered with shift $z$ the word $u$ is entered with shift $z-2$ and an $\encodeDeletion$ gets appended as shown in Figure~\ref{fig:reading objects when number of symbols read is upperbound} (iv).
So using the $b$ symbols in each object we control the tracks read in each $u$ within the object so that the correct appendant gets appended when the object is read. The details of reading objects and $u$ subwords are given in Section~\ref{sec:Reading objects and $u$ subwords}.

\subsubsection{Encoding the marked appendant of $\CTS'$ in the shift of $\TSTC$}\label{sec:Encoding the marked appendant} 
From Table~\ref{tab:equalities} the shift change when reading an $\encodeZero$ or an $\encodeOne$ object is $z_1$. So from Lemma~\ref{lem:shiftChange}, when an $\encodeZero$ or an $\encodeOne$ object is entered with shift $(z_1m+z_2d)\bmod\beta$, the next object immediately to its right is entered with shift $(z_1(m+1)+z_2d)\bmod\beta$ as shown in (i) and (ii) of Figures~\ref{fig:reading objects when number of symbols read is upperbound} and~\ref{fig:reading objects when number of symbols read is lowerbound}.
This shift change of $z_1$ simulates that the marked appendant changes from $\alpha_{m}$ to $\alpha_{((m+1)\bmod (3x-2))}$ (the length of the program for $\CTS'$ is $3x-2$). 
If $\CTSp$ is at the marked appendant $\alpha_m$ and then reads $3x-2$ symbols, it traverses its entire circular program and returns to appendant $\alpha_m$. 
Notice from Table~\ref{tab:equalities} that $z_1(3x-2)=\beta$, and so if we read $3x-2$ of the $\encodeZero$ and $\encodeOne$ objects, then the total shift change is $0=z_1(3x-2)\bmod\beta$.
Since the shift change value is $0$, the encoding of the marked appendant remains unchanged after reading $3x-2$ of the $\encodeZero$ and $\encodeOne$ objects, correctly simulating a traversal of the entire circular program of $\CTS'$.
In Lemma~\ref{lem:Encode appendant} it is proved that if $m_i\neq m_j$ then $((z_1m_i+z_2d_i)\bmod\beta)\neq ((z_1m_j+z_2d_j)\bmod\beta)$ for all $0\leqslant d_i,d_j<3x+1$.
This shows that each shift value $z=((z_1m+z_2d)\bmod\beta)$ encodes one and only one appendant $\alpha_{m}$. So reading $\encodeDeletionPrime$, with its shift change of $z_2$, moves from one shift value that encodes $\alpha_{m}$ to another shift value that also encodes $\alpha_{m}$ (as shown in (iii) of Figures~\ref{fig:reading objects when number of symbols read is upperbound} and~\ref{fig:reading objects when number of symbols read is lowerbound}). In other words, reading $\encodeDeletionPrime$ does not change the value of the appendant encoded in the shift. Finally, since $\encodeDeletion$ has a shift change value of 0 it too does not change the appendant encoded in the shift.
\begin{figure}
\setlength{\unitlength}{.6cm}
\begin{picture}(10,2.2)
\put(0,1.1){(i)}
\put(1,1.1){$\encodeOneTrack\;\rightarrow\;\;\encodeOneSequence$} 
\put(0,0){(ii)}
\put(1,0){$\encodeZeroTrack\;\rightarrow\;\;\encodeZeroSequence$} 
\put(17.3,1.1){(iii)}
\put(18.5,1.1){$\encodeDeletionTrack\;\rightarrow\;\;\encodeDeletionSequence$} 
\put(14.6,0){(iv)}
\put(15.8,0){$\encodeDeletionPrimeTrack\;\rightarrow\;\;\encodeDeletionPrimeSequence$}
\end{picture}
\caption{Sequence of symbols read to append $\encodeOne$ (i), $\encodeZero$ (ii), $\encodeDeletion$ (iii), and $\encodeDeletionPrime$ (iv). Rules $b\rightarrow b$ or $c\rightarrow u$ are applied to each symbol in sequence on the left to give the object it appends on the~right.}\label{fig:tracks for appending each object} 
\end{figure}

\subsubsection{Reading objects and defining the word \emph{u}}\label{sec:Reading objects and $u$ subwords}
The word $u$ is defined via Tables~\ref{tab:Tracks in u for Case 1} to~\ref{tab:Tracks in u for Case 2b} such that when each object is read it appends the correct appendant as shown in Figures~\ref{fig:reading objects when number of symbols read is upperbound} and~\ref{fig:reading objects when number of symbols read is lowerbound}. 
We will take the case of reading an $\encodeOne$ entered with shift $z<\beta-z_1$ and show that it appends $\tne{\alpha_m}$ as illustrated in Figure~\ref{fig:reading objects when number of symbols read is upperbound} (i). We will then explain how the method used to verify this case can be applied to verify the remaining cases in~Figures~\ref{fig:reading objects when number of symbols read is upperbound}~and~\ref{fig:reading objects when number of symbols read is lowerbound}.

Here we show that when $\encodeOne$ is entered with shift $z<\beta-z_1$ the sequence of symbols read in the $u$ subwords of $\encodeOne$ append $\tne{\alpha_m}=\tne{\sigma_1}\tne{\sigma_2}\ldots\tne{\sigma_v}\encodeDeletion^{x-v+1}$. 
Recall that the sequence of symbols (or track) read in a word depends on the shift with which the word is entered (see Figure~\ref{fig:u track and encodedOne track} (i)). 
From Lemma~\ref{lem:shiftChange}, when $\encodeOne=\encodeOneSequence$ is entered with shift $z$ then the leftmost $b^{10}$ causes the leftmost $u$ to be entered with shift $(z-10)\bmod\beta$, and because each $u$ has a shift change of $3x$, following the $ubb$ subword the second $u$ is entered with shift $(z+(3x-2)-10)\bmod\beta$, the third $u$ with shift $(z+2(3x-2)-10)\bmod\beta$, and so on (as shown in Figure~\ref{fig:u track and encodedOne track} (ii)). 
Now we define the track read in each $u$ in Figure~\ref{fig:u track and encodedOne track} (ii) so that it appends a single object from $\{\encodeZero,\encodeOne,\encodeDeletion,\encodeDeletionPrime\}$ at the right end of the dataword. 
The tracks that append each object are given in Figure~\ref{fig:tracks for appending each object}. For example in Figure~\ref{fig:tracks for appending each object} (iii) we see that applying the rules $b\rightarrow b$ and $c\rightarrow u$ to the sequence $\encodeDeletionTrack$ appends the object $\encodeDeletion=\encodeDeletionSequence$.
Note from Figure~\ref{fig:tracks for appending each object} that the number of symbols read to append each object is either of length $3x+1$ or $3x$.
Now note from Table~\ref{tab:equalities} and Lemma~\ref{lem:read upper or lower bound number of symbols} that when a $u$ word is read we read either $3x+1$ or $3x$ symbols, and so each $u$ subword that is read can append a single object from $\{\encodeDeletion,\encodeDeletionPrime,\encodeZero,\encodeOne\}$. For example, in Figure~\ref{fig:u track and encodedOne track} (i) if we wish $u$ to append $\encodeDeletion=\encodeDeletionSequence$ when entered with shift $s$ then we define $\track{s}{u}=\encodeDeletionTrack$.

To append the sequence of objects $\tne{\alpha_m}=\tne{\sigma_1}\tne{\sigma_2}\ldots\tne{\sigma_v}\encodeDeletion^{x-v+1}$ when reading an $\encodeOne$, the track read in the $i+1^{\textrm{th}}$ $u$ from the left in $\encodeOne$ appends the $i+1^{\textrm{th}}$ object from the left in $\tne{\alpha_m}$, where $0\leqslant i\leqslant x$. 
Thus in Figure~\ref{fig:u track and encodedOne track} (ii), for $0\leqslant i< v$ if $\sigma_{i+1}=0$, then track $\track{z+i(3x-2)-10}{u}=\encodeZeroTrack$ is read causing the word $\encodeZero=\encodeZeroSequence$ to be appended, and if $\sigma_{i+1}=1$, then track $\track{z+i(3x-2)-10}{u}=\encodeOneTrack$ is read causing the word $\encodeOne=\encodeOneSequence$ to be appended. 
There is an exception when $i=0$ and $z=0$ as we have $\beta-10=(z+i(3x-2)-10)\bmod\beta$. 
In this special case, tracks have the from $\track{\beta-10}{u}=\encodeZeroTrackSpecialCaseA$ or $\track{\beta-10}{u}=\encodeOneTrackSpecialCaseA$ with one less $b$ than usual. 
From Lemma~\ref{lem:read upper or lower bound number of symbols} and Table~\ref{tab:equalities}, when $u$ is entered with shift $\beta-10$ only $\lfloor\frac{|u|}{\beta}\rfloor=3x$ (instead of $3x+1$) symbols are read. When $u$ is entered with shift $\beta-10$ then $z=0$, and so the leftmost $b$ in the $\encodeOne$ is read and provides the first $b$ in the sequence $\encodeZeroTrack$ that prints an $\encodeZero$, or provides the first $b$ in the sequence $\encodeOneTrack$ that prints an $\encodeOne$. 
For $v\leqslant i< \frac{x}{2}-7$, each track $\track{z+i(3x-2)-10}{u}=\encodeDeletionTrack$ appends the word $\encodeDeletion$, and for $\frac{x}{2}-7\leqslant i< x$ each track $\track{z+3xi-x+4}{u}=\encodeDeletionTrack$ also appends the word $\encodeDeletion$. 
Track $\track{z+3x^2-x+2}{u}=\encodeDeletionTrack$ appends the word $\encodeDeletion$. 
There is an exception if we have and $z=\beta-z_1-x$ which gives $\beta-3x+2=(z+3x^2-x+2)$. 
In this special case, tracks have the from $\track{\beta-3x+2}{u}=\encodeDeletionTrackCaseB$ with one less $b$ than usual. From Lemma~\ref{lem:read upper or lower bound number of symbols} and Table~\ref{tab:equalities}, when $u$ is entered with shift $\beta-3x+2$ only $\lfloor\frac{|u|}{\beta}\rfloor=3x$ (instead of $3x+1$) symbols are read. 
When $u$ is entered with shift $\beta-3x+2$ a $b$ from the $b^{x+2}$ at the right end of the $\encodeOne$ is read and provides the last $b$ in the sequence $\encodeDeletionTrack$ that appends an $\encodeDeletion$.
The $u$ tracks in this paragraph show that when an $\encodeOne$ is entered with shift $z<\beta-z_1$ the encoding of $\tne{\alpha_m}$ is appended at the right end of the dataword as shown in Figure~\ref{fig:reading objects when number of symbols read is upperbound} (i). The $u$ tracks given above have the same values as the bottom eight $u$ tracks in Table~\ref{tab:Tracks in u for Case 1}.

\begin{figure}
\setlength{\unitlength}{.6cm}
\begin{picture}(10,4.5)
\put(0,3.5){(i)}
\put(1.3,3.6){$u=u_0\ldots u_{s-1}\pmb{u_s}u_{s+1}\ldots u_{s+\beta-1}\pmb{u_{s+\beta}}u_{s+\beta+1}\ldots u_{|u|-1}$} 
\put(18.2,3.6){$\track{s}{u}=u_su_{s+\beta}u_{s+2\beta}\ldots u_{s+(3x+1)\beta}$} 
\put(0,2.4){(ii)}
\put(1.3,2.4){\fontsize{10.5}{1.2}$b^{10}\;\,\track{(z-10)\bmod\beta}{u}\;\;\,b^2\;\;\track{(z+(3x-2)-10)\bmod\beta}{u}\;\;\,b^2\;\;\track{(z+2(3x-2)-10)\bmod\beta}{u}\;\;\ldots\;\; \track{(z+(\frac{x}{2}-8)(3x-2)-10)\bmod\beta}{u}\;\;\,b^2$}
\put(1.4,1.2){\fontsize{10.5}{1.2}$\track{(z+3x(\frac{x}{2}-7)-x+4)\bmod\beta}{u}\;\;\;\;\track{(z+3x(\frac{x}{2}-6)-x+4)\bmod\beta}{u}\;\;\;\;\track{(z+3x(\frac{x}{2}-5)-x+4)\bmod\beta}{u}\;\;\ldots\;\;\track{(z+3x(x-1)-x+4)\bmod\beta}{u}$}
\put(1.3,-0.1){\fontsize{10.5}{1.2}$b^2\;\;\track{(z+3x^2-x+2)\bmod\beta}{u}\;\;\;b^{x+2}$}
\end{picture}
\caption{(i) The word $u$ where $u_i\in\{b,c\}$. The symbols given in bold are read when $u$ is entered shift value $s$ and this bold symbol sequence defines track $\protect\track{s}{u}$. (ii) The tracks read in each $u$ when $\encodeOne=\encodeOneSequence$ is entered with shift $z$.
The $b$ symbols in (ii) are not part of~the~$u$~tracks.}\label{fig:u track and encodedOne track}
\end{figure}

\setlength{\extrarowheight}{5pt}

\begin{table*}[!ht]
\centering\vspace{-1ex}
\begin{tabular}{@{}c|l|l}
   Object track  & Tracks read in $u$ & Values for $z$, $i$ and $\sigma_{i+1}$ \\ \hline
  & $\track{\beta-2}{u}=\encodeDeletionTrackCaseA\quad$ & $i=0$,\;\;\;\; $z=0$ \\ \cline{2-3}
 $\track{z}{\encodeDeletion}=\encodeDeletionTrack$ & $\track{z-2}{u}=\encodeDeletionTrack\quad$ & $i=0$,\;\;\;\; $0< z-2<\beta-3x$\\  \cline{2-3}
  & $\track{z-2}{u}=b^{2}cb^{3x-3}\quad$ & $i=0$,\;\;\; $\beta-3x \leqslant z-2<\beta-2$ \\  
 \Xhline{2\arrayrulewidth}
 
    & $\track{\beta-4}u=\encodeDeletionTrackCaseA\quad$  & $i=0$,\;\;\;\; $z=0$\\  \cline{2-3}
& $\track{z-4}{u}=\encodeDeletionTrack\quad$  & $i=0$,\;\;\;\;$0< z<\beta-z_2$ \\  \cline{2-3}
    $\track{z}{\encodeDeletionPrime}=(\encodeDeletionTrack)^{x}$ &  $\track{z+3xi-6}{u}=\encodeDeletionTrack$ &$1\leqslant i< x-1$,\;\;\;\; $0\leqslant z<\beta-z_2$ \\ \cline{2-3}
&  $\track{z+3x(x-1)-8}{u}=\encodeDeletionTrack$ & $i= x-1$,\;\;\;\; $0\leqslant z\leqslant\beta-z_2-2x$ \\ \cline{2-3}
&  $\track{\beta-2x-8}{u}=\encodeDeletionTrackCaseB$ & $i= x-1$,\;\;\;\;$z=\beta-z_2-x$ \\ 
\Xhline{2\arrayrulewidth}
    & $\track{\beta-4}{u}=\encodeDeletionTrackCaseA\quad$  & $i=0$,\;\;\;\; $z=0$\\  \cline{2-3}
& $\track{z-4}{u}=\encodeDeletionTrack\quad$  & $i=0$,\;\;\;\;$0< z<\beta-z_1$ \\  \cline{2-3}
$\track{z}{\encodeZero}=(\encodeDeletionTrack)^{x+1}$&  $\track{z+3xi-6}{u}=\encodeDeletionTrack$ &$1\leqslant i< x$,\;\;\;\;  $0\leqslant z<\beta-z_1$ \\ \cline{2-3}
&  $\track{z+3x^2-8}{u}=\encodeDeletionTrack$ & $i= x$,\;\;\;\;$0\leqslant z\leqslant\beta-z_1-2x$ \\ \cline{2-3}
&  $\track{\beta-2x-8}{u}=\encodeDeletionTrackCaseB$ & $i= x$,\;\;\;\;$z=\beta-z_1-x$ \\  \cline{2-3}
\Xhline{2\arrayrulewidth}

   & $\track{\beta-10}{u}=\encodeZeroTrackSpecialCaseA$ & $i=0$,\;\;\;\; $z=0$,\;\;\;\; $\sigma_1=0$ \\ \cline{2-3}
  &  $\track{\beta-10}{u}=\encodeOneTrackSpecialCaseA$& $i=0$,\;\;\;\; $z=0$,\;\;\;\; $\sigma_1=1$ \\ \cline{2-3}\vspace{-3pt}
  & $\track{z+i(3x-2)-10}{u}=\encodeZeroTrack$ &  $0\leqslant i< v$,\;\;\;\; $0\leqslant z<\beta-z_1$,\\  
 & &$(i,z)\neq(0,0)$,\;\;\;\; $\sigma_{i+1}=0$\\ \cline{2-3}\vspace{-6pt}
 $\track{z}{\encodeOne}=\langle\alpha_m\rangle$ &$\track{z+i(3x-2)-10}{u}=\encodeOneTrack$  &  $0\leqslant i< v$,\;\;\;\; $0\leqslant z<\beta-z_1$,\\
Equation~\eqref{eq:alpha m} &&$(i,z)\neq(0,0)$,\;\;\;\; $\sigma_{i+1}=1$\\ 
\cline{2-3}
 & $\track{z+i(3x-2)-10}{u}=\encodeDeletionTrack$ &  $v\leqslant i< \frac{x}{2}-7$,\;\;\;\; $0\leqslant z<\beta-z_1$\\ \cline{2-3}
& $\track{z+3xi-x+4}{u}=\encodeDeletionTrack$ &  $\frac{x}{2}-7\leqslant i< x$,\;\;\;\; $0\leqslant z<\beta-z_1$\\ \cline{2-3}
 & $\track{z+3x^{2}-x+2}{u}=\encodeDeletionTrack$ & $i=x$,\;\;\;\; $0\leqslant z\leqslant\beta-z_1-2x$\\ \cline{2-3}
  & $\track{\beta-3x+2}{u}=\encodeDeletionTrackCaseB$ & $i=x$,\;\;\;\; $z=\beta-z_1-x$\\
\hline
\end{tabular}
\caption{Tracks read in each object. Here $\encodeDeletion$ is entered with shift $z<\beta$, $\encodeDeletionPrime$ is entered with a shift $z<\beta-z_2$, $\encodeZero$ and $\encodeOne$ are entered with a shift $z<\beta-z_1$. The values $z$, $z_1$, and $z_2$ are given in Table~\ref{tab:equalities}, and $\tne{\alpha_m}$, $\sigma_{i+1}$ and $v$ are given in Equation~\eqref{eq:alpha m}. The value $i$ indexes the position of the $u$ subword within the object being read (see Figure~\ref{fig:u track and encodedOne track} (ii)).  Here the $\bmod\;\beta$ is dropped from the underscripts in $u$ tracks as all of the underscript terms above are $0\leqslant$ and $<\beta$.}\label{tab:Tracks in u for Case 1}
\end{table*}

\setlength{\extrarowheight}{6pt}
\begin{table*}[!ht]
\centering
\begin{tabular}{@{}c|l|l@{\;}}
 Object track  & Tracks read in $u$ & Values for $z$, $i$, $j$, and $\sigma_{i+1}$ 
 \\ \hline
					& $\track{z-4}{u}=\encodeDeletionTrack$& $i=0$,\;\;\;\; $\beta-z_2-4 \leqslant z-4<\beta-3x$  
\\ \cline{2-3}
					& $\track{z-4}{u}=\encodeDeletionPrimeTrack$& $j=0$,\;\;\;\; $\beta-3x\leqslant z-4<\beta-4$, 
\\ \cline{2-3}
$\track{z}{\encodeDeletionPrime}=$	& $\track{(z+3xi-6)\bmod\beta}{u}=\encodeDeletionTrack$ & $1\leqslant i<x-1$,\;\;\;\; $i\neq j$
\\\cline{2-3}\vspace{-4pt}
$(\encodeDeletionTrack)^{j}$		& $\track{z+3xj-6}{u}=\encodeDeletionPrimeTrack$& $1\leqslant j<x-1$ \\
$\encodeDeletionPrimeTrack$		&&  $\beta-3x\leqslant z+3xj-6<\beta$
\\ \cline{2-3} 
$(\encodeDeletionTrack)^{x-j-1}$	& $\track{(z+3x(x-1)-8)\bmod\beta}{u}=\encodeDeletionTrack$&  $i=x-1$,\;\;\;\;$\beta< z+3x(x-1)-8$ 
\\\cline{2-3}\vspace{-4pt}
					& $\track{z+3x(x-1)-8}{u}=\encodeDeletionPrimeTrack$& $j=x-1$
\\
					&&$\beta-2x< z+3x(x-1)-8<\beta$
\\
\Xhline{2\arrayrulewidth}

					& $\track{z-4}{u}=\encodeDeletionTrack$& $i=0$,\;\;\;\; $\beta-z_1-4 <z-4<\beta-3x$ 
\\  \cline{2-3}
					& $\track{z-4}{u}=\encodeDeletionPrimeTrack$& $j=0$,\;\;\;\;  $\beta-3x\leqslant z-4<\beta-4$
\\\cline{2-3}
$\track{z}{\encodeZero}=$		& $\track{(z+3xi-6)\bmod\beta}{u}=\encodeDeletionTrack$ & $1\leqslant i<x$,\;\;\;\; $i\neq j$
\\\cline{2-3} \vspace{-4pt}
$(\encodeDeletionTrack)^{j}$		& $\track{z+3xj-6}{u}=\encodeDeletionPrimeTrack$& $1\leqslant j<x$ 
\\ 
$ \encodeDeletionPrimeTrack$		& &$\beta-3x\leqslant z+3xj-6<\beta$
\\\cline{2-3}
$(\encodeDeletionTrack)^{x-j}$		& $\track{(z+3x^2-8)\bmod\beta}{u}=\encodeDeletionTrack$&  $i=x$,\;\;\;\; $\beta< z+3x^2-8$
\\\cline{2-3}
					& $\track{(z+3x^2-8)}{u}=\encodeDeletionPrimeTrack$& $j=x$,\;\;\;\; $\beta-2x< z+3x^2-8<\beta$
\\
\hline

\end{tabular}
\caption{Tracks read in $\encodeDeletionPrime$ when entered with shift $z\geqslant\beta-z_2$, and tracks read in $\encodeZero$ when entered with shift $z\geqslant\beta-z_1$. The values $z$, $z_1$, and $z_2$ are given in Table~\ref{tab:equalities}. The value $i$ indexes the position of the $u$ subword within the object being read (see Figure~\ref{fig:u track and encodedOne track} (ii)). The value $j$ gives the index of the $u$ subword that appends $\encodeDeletionPrime$ (see Figure~\ref{fig:reading objects when number of symbols read is lowerbound}). The $\bmod\;\beta$ is dropped from underscripts where the term is  $<\beta$.}\label{tab:Tracks in u for Case 2a}
\end{table*}

\setlength{\extrarowheight}{7pt}
\begin{table*}[!ht]
\centering
\begin{tabular}{@{}c|l|l@{\;}}
 Object track  & Tracks read in $u$ & Values for $z$, $i$, $j$, and $\sigma_{i+1}$ 
 \\ \hline
							& $\track{z+i(3x-2)-10}{u}=\encodeZeroTrack$  &  $0\leqslant i< j$,\;\;\;\;  $i<v$,\;\;\;\;\;$\sigma_{i+1}=0$
\\ \cline{2-3}
							& $\track{z+i(3x-2)-10}{u}=\encodeOneTrack$ &   $0\leqslant i< j$,\;\;\;\;  $i<v$,\;\;\;\; $\sigma_{i+1}=1$
 \\ \cline{2-3}
 							& $\track{(z+i(3x-2)-10)\bmod\beta}{u}=\encodeZeroTrack$ &   $j< i\leqslant v$,\;\;\;\; $\sigma_{i}=0$
\\ \cline{2-3}
							& $\track{(z+i(3x-2)-10)\bmod\beta}{u}=$ &   $j< i\leqslant \underset{}{v}$,\;\;\;\; $\sigma_{i}=1$
							\\
							&$\;\;\;\qquad\qquad\qquad\encodeOneTrack$
							\\  \cline{2-3}
$\track{z}{\encodeOne}=\langle\alpha'_{m}\rangle$	& $\track{z+i(3x-2)-10}{u}=\encodeDeletionTrack$ &  $i<\frac{x}{2}-7$,\;\;\;\;  $ i<j$,\;\;\;\;  $i\geqslant v$  
\\ \cline{2-3}
Equations						& $\track{(z+i(3x-2)-10)\bmod\beta}{u}=\encodeDeletionTrack$ &  $i<\frac{x}{2}-7$,\;\;\;\;  $ i>j$,\;\;\;\;  $i>v$
  \\\cline{2-3}\vspace{-4pt}
 \eqref{eq:alphaPrime m1} and~\eqref{eq:alphaPrime m2} & $\track{z+j(3x-2)-10}{u}=\encodeDeletionPrimeTrack$ &    $j<\frac{x}{2}-7$
 \\
							&&$\beta-3x\leqslant z+j(3x-2)-10<\beta-10$
 \\\cline{2-3}
							& $\track{(z+3xi-x+4)\bmod\beta}{u}=\encodeDeletionTrack$ &  $\frac{x}{2}-7\leqslant i<x$,\;\;\;\; $i\neq j$ \\
							\cline{2-3}\vspace{-4pt}
							& $\track{z+3xj-x+4}{u}=\encodeDeletionPrimeTrack$ 	&  $\frac{x}{2}-7\leqslant j < x$
   \\ 
							&&$\beta-3x\leqslant z+3xj-x+4<\beta$
   \\\cline{2-3}
							& $\track{(z+3x^2-x+2)\bmod\beta}{u}=\encodeDeletionTrack$ & $i=x$,\;\;\;\; $j<x$\\\cline{2-3}\vspace{-4pt}
							& $\track{z+3x^2-x+2}{u}=\encodeDeletionPrimeTrack$ &  $j=x$\\
							&&$  \beta-2x< z+3x^2-x+2<\beta$\\ 
\hline
\end{tabular}
\caption{Track read in $\encodeOne$, when entered with shift $z>\beta-z_1$. The values $z$ and $z_1$ are given in Table~\ref{tab:equalities}, and $\tne{\alpha_m}$ and $v$ are given in Equations~\eqref{eq:alphaPrime m1} and~\eqref{eq:alphaPrime m2}. The value $i$ indexes the position of the $u$ subword within the object being read (see Figure~\ref{fig:u track and encodedOne track} (ii)). The value $j$ gives the index of the $u$ subword that appends $\encodeDeletionPrime$ (see Figure~\ref{fig:reading objects when number of symbols read is lowerbound}). The $\bmod\;\beta$ is dropped from underscripts where the term is  $<\beta$.}\label{tab:Tracks in u for Case 2b}
\end{table*}

Using the same method as in the previous two paragraphs one can show that the $u$ tracks given for each of the objects $\encodeDeletion$, $\encodeZero$ and $\encodeDeletionPrime$ in Table~\ref{tab:Tracks in u for Case 1} will cause the correct appendant to be appended as shown in Figure~\ref{fig:reading objects when number of symbols read is upperbound}. 
Note that once the shift values for the $u$ tracks have been determined (as in Figure~\ref{fig:u track and encodedOne track} (ii)) we can use these shift values to determine the special cases. 
When $u$ is entered with a shift $\geqslant\beta-3x$, then from Lemma~\ref{lem:read upper or lower bound number of symbols} only $3x$ (instead of $3x+1$) symbols are read and we have a special case where the object track is missing a single $b$. 
For the special cases in Table~\ref{tab:Tracks in u for Case 1} (rows 1, 3, 4, 8, 9, 13, 14, 15 and 21) the missing $b$ needed to complete the object track is provided by reading a $b$ in a sequence of $b$ symbols to the left or right of the $u$ that is entered with shift $\geqslant\beta-3x$. 
For example, from row 4 of Table~\ref{tab:Tracks in u for Case 1} when $\encodeDeletionPrime=\encodeDeletionPrimeSequence$ is entered with shift $z=0$ we read track $\track{\beta-4}{u}=\encodeDeletionTrackCaseA$ in the leftmost $u$ of $\encodeDeletionPrime$, and since $z=0$ we also read the leftmost $b$ in $\encodeDeletion$ which provides the leftmost $b$ needed to complete the track $\encodeDeletionTrack$ that appends an $\encodeDeletion$. 
Note that in Table~\ref{tab:Tracks in u for Case 1} there are no special cases (i.e.~tracks of length $3x$) for $1\leqslant i<x-1$ when reading $\encodeDeletionPrime$, and for $1\leqslant i<x$ when reading $\encodeZero$ and $\encodeOne$. 
This is because when $\encodeDeletionPrime$ is entered with shifts $<\beta-z_2$, and when $\encodeZero$ and $\encodeOne$ are entered with shifts $<\beta-z_1$, it is not possible to enter the $u$ subwords at these positions with shifts $\geqslant\beta-3x$.

The method used earlier in this section can also be used to demonstrate that the $u$ tracks for each object in Tables~\ref{tab:Tracks in u for Case 2a} and~\ref{tab:Tracks in u for Case 2b} will cause the appendants shown in Figure~\ref{fig:reading objects when number of symbols read is lowerbound} to be appended when each object is read. 
Note from the captions of Figures~\ref{fig:reading objects when number of symbols read is upperbound} and~\ref{fig:reading objects when number of symbols read is lowerbound} that the shift values differ between the two figures. 
From Table~\ref{tab:equalities} and Lemma~\ref{lem:read upper or lower bound number of symbols}, the number of symbols read in each object in  Figure~\ref{fig:reading objects when number of symbols read is upperbound} is $y\in\{\lceil\frac{|\encodeOne|}{\beta}\rceil,\lceil\frac{|\encodeZero|}{\beta}\rceil,\lceil\frac{|\encodeDeletionPrime|}{\beta}\rceil,\frac{|\encodeDeletion|}{\beta}\}$ and the number of symbols read in each object in  Figure~\ref{fig:reading objects when number of symbols read is lowerbound} is $y-1\in\{\lfloor\frac{|\encodeOne|}{\beta}\rfloor,\lfloor\frac{|\encodeZero|}{\beta}\rfloor,\lfloor\frac{|\encodeDeletionPrime|}{\beta}\rfloor\}$ (there is only one case for reading $\encodeDeletion$ since $\lceil\frac{|\encodeDeletion|}{\beta}\rceil=\lfloor\frac{|\encodeDeletion|}{\beta}\rfloor$). 
Note from Figure~\ref{fig:tracks for appending each object}, to append an $\encodeDeletion$ object we read a symbol sequence of length $3x+1$, and to append an $\encodeDeletionPrime$ object we read a symbol sequence of length $3x$. 
So when we read $y-1$ symbols (instead of $y$) in an object as shown in Figure~\ref{fig:reading objects when number of symbols read is lowerbound}, we include an $\encodeDeletionPrime$ object instead of one of the $\encodeDeletion$ objects as this gives an object track that is one symbol shorter than the tracks read in Figure~\ref{fig:reading objects when number of symbols read is upperbound}. 
Recall that only $3x$ symbols are read when $u$ is entered with a shift $\geqslant\beta-3x$, and so the location of the $\encodeDeletionPrime$ object in the sequence of objects that are appended depends on which $u$ is entered with a shift $\geqslant\beta-3x$. 
In Tables~\ref{tab:Tracks in u for Case 2a} and~\ref{tab:Tracks in u for Case 2b} we introduce the variable $j$ to denote the position of the $u$ word within the object that appends the $\encodeDeletionPrime$ object. 
In rows 2, 6, 8 and 12 in Table~\ref{tab:Tracks in u for Case 2a} and rows 7 and 11 in Table~\ref{tab:Tracks in u for Case 2b} the range of shift values for which $\encodeDeletionPrime$ gets appended by $\track{s}{u}$ is less than the usual range of $\beta-3x\leqslant s<\beta$. 
The reason for this is that the object being read is entered with a shift $z$ that does not give all values in the range $\beta-3x\leqslant s<\beta$ for these cases. 
For example, in row 12 of Table~\ref{tab:Tracks in u for Case 2a} we have $\beta-2x< z+3x^2-8<\beta$ because it is not possible to have $z+3x^2-8\leqslant\beta-2x$ when $\encodeZero$ is entered with shift $z\geqslant\beta-z_1$.

In this section we provided a method for showing that $u$ is defined via Tables~\ref{tab:Tracks in u for Case 1} to~\ref{tab:Tracks in u for Case 2b} such that each object appends the correct sequence of objects as shown in Figures~\ref{fig:reading objects when number of symbols read is upperbound} and~\ref{fig:reading objects when number of symbols read is lowerbound}. In Lemma~\ref{lem:Correctness of u} we prove the correctness of $u$ by show that we have not assigned more than one value to the same track in the word $u$.

\subsubsection{Complexity analysis}
We give the time analysis for $\TSTC$ simulating the cyclic tag systems $C$ that runs in time $t$. 
During the simulation, for every $3x-2$ objects from $\{\encodeZero,\encodeOne\}$ that are read, we enter one of these objects with shift $\geqslant\beta-z_1$ (this is because $z_1(3x-2)=\beta$ and $\{\encodeZero,\encodeOne\}$ have a shift change of $z_1$). 
From Figure~\ref{fig:reading objects when number of symbols read is lowerbound} (i) and (ii), if we enter an $\encodeZero$ or an $\encodeOne$ object with shift $\geqslant\beta-z_1$, then there is a single $\encodeDeletionPrime$ object in the sequence of objects that are appended. So after reading $t$ objects from $\{\encodeZero,\encodeOne\}$ to simulate $t$ steps of $C$, we have $O(t)$ of the $\encodeDeletionPrime$ objects in the dataword of $\TSTC$. 
For each object read from $\{\encodeDeletionPrime,\encodeZero,\encodeOne\}$, a constant number (independent of the input) of the $\encodeDeletion$ objects are appended, and so we have $O(t)$ of the $\encodeDeletion$ objects in the dataword of $\TSTC$. 
There are $O(t)$ objects from $\{\encodeZero,\encodeOne\}$ in the dataword, and so the space used by $\TSTC$ is $O(t)$. 
Between each pair of objects $\tne{w_1},\tne{w_2}\in\{\encodeZero,\encodeOne\}$ that encode adjacent symbols in the dataword of $\CTS$ there are $O(t)$ of the $\encodeDeletion$ and $\encodeDeletionPrime$ objects. 
So the word $\tne{w_1}\{\encodeDeletion,\encodeDeletionPrime\}^\ast$ that is read to simulate a computation step, as described in the second last paragraph Section~\ref{sec:Algorithm overview}, has $O(t)$ objects.
From Figures~\ref{fig:reading objects when number of symbols read is upperbound} and~\ref{fig:reading objects when number of symbols read is lowerbound} it takes a constant number of steps to read each object, and thus reading $\tne{w_1}\{\encodeDeletion,\encodeDeletionPrime\}^\ast$ to simulate a single computation step takes time $O(t)$.
So, $\TSTC$ simulates a single step of $\CTS$ in time $O(t)$, and $t$ steps of $C$ in time $O(t^2)$.

\subsection{Correctness of $\TSTC$}
Note that in addition to the proof of correctness given here, $\TSTC$ was implement in software and tested extensively. Below, the correctness of $\TSTC$ is proved by showing that it correctly simulates an arbitrary computation step of the cyclic tag system $\CTSp$. In the third, fourth and fifth paragraphs of Section~\ref{sec:Algorithm overview} is an overview of how $\TSTC$ reads a word of the form $\tne{w_1,z}\{\encodeDeletion,\encodeDeletionPrime\}^\ast$ (where $w_1\in\{1,0\}$) to simulate a computation step of $\CTSp$. Lemma~\ref{lem:simulating a computation step of C} shows that $\TSTC$ reads $\tne{w_1,z}\{\encodeDeletion,\encodeDeletionPrime\}^\ast$ to correctly simulate an arbitrary computation step of $\CTSp$. The $\TSTC$ dataword immediately before the simulated computation step is given by Equation~\eqref{eq:TS dataword before simulated computation step} and the $\TSTC$ dataword immediately after the simulated computation step is given by Equation~\eqref{eq:TS dataword after simulated computation step}. In Equation~\eqref{eq:TS dataword after simulated computation step} the next encoded symbol to be read, $\tne{w_2}$, is at the left end of the dataword, and so after the simulated computation step  the dataword has the correct form to begin the simulation of the next computation step. It follows that $\TSTC$ correctly simulates the computation $\CTSp$. Recall that cyclic tag systems end their computation by entering a repeating sequence of configurations. While $\TSTC$ correctly simulates this repeating sequence of configurations, $\TSTC$ itself does not enter a repeating sequence as the number of garbage objects in the dataword increases with each simulated computation step.

In Lemma~\ref{lem:simulating a computation step of C} the objects $\encodeOne$, $\encodeZero$, $\encodeDeletion$, and $\encodeDeletionPrime$ are defined in Section~\ref{sec:Encoding used by tag system} and the values $z_1$, $z_2$, $m$ and $d$ can be found in Table~\ref{tab:equalities}. Equations~\eqref{eq:TS dataword before simulated computation step} and~\eqref{eq:TS dataword after simulated computation step} encode arbitrary datawords of $C'$ in the manner described by Equation~\eqref{eq:encoding CTS configuration}.

\begin{lemma}[$\TSTC$ simulates an arbitrary computation step of $\CTSp$]\label{lem:simulating a computation step of C}
Given a dataword of the form
\begin{equation}\label{eq:TS dataword before simulated computation step}
\tne{w_1,z}\{\encodeDeletion,\encodeDeletionPrime\}^\ast\tne{w_2}\{\encodeDeletion,\encodeDeletionPrime\}^\ast\tne{w_3}\ldots \{\encodeDeletion,\encodeDeletionPrime\}^\ast\tne{w_l}\{\encodeDeletion,\encodeDeletionPrime\}^\ast 
\end{equation}
where $z=(z_1m+z_2d)\bmod\beta$ and $w_i\in\{0,1\}$, then a single round of $\TSTC$ on the word $\tne{w_1,z}\{\encodeDeletion,\encodeDeletionPrime\}^\ast$ gives a dataword of the form
\begin{equation}\label{eq:TS dataword after simulated computation step}
\tne{w_2,z'}\{\encodeDeletion,\encodeDeletionPrime\}^\ast\tne{w_3}\ldots \{\encodeDeletion,\encodeDeletionPrime\}^\ast\tne{w_l}\{\encodeDeletion,\encodeDeletionPrime\}^\ast w' \{\encodeDeletion,\encodeDeletionPrime\}^\ast
\end{equation}
where $z'=(z_1(m+1)+z_2d')\bmod\beta$ and $0\leqslant d'<3x+1$, and
\begin{equation}\label{eq:values for w'}
w'=
\begin{cases} \tne{\alpha_m} &\quad\textrm{if}\quad  \tne{w_1}=\encodeOne\quad \textrm{and}\quad z<\beta-z_1\\
 \tne{\alpha_m'} &\quad\textrm{if}\quad  \tne{w_1}=\encodeOne\quad \textrm{and}\quad z\geqslant\beta-z_1\\
\encodeDeletion^{x+1} &\quad\textrm{if}\quad \tne{w_1}=\encodeZero\quad \textrm{and}\quad z<\beta-z_1\\
\encodeDeletion^{j}\encodeDeletionPrime\encodeDeletion^{x-j}  &\quad\textrm{if}\quad  \tne{w_1}=\encodeZero\quad \textrm{and}\quad z\geqslant\beta-z_1
\end{cases}
\end{equation}
\end{lemma}
\begin{proof}
We begin by showing that to prove this lemma it is sufficient to verify that $\TSTC$ behaves as described in Figures~\ref{fig:reading objects when number of symbols read is upperbound} and~\ref{fig:reading objects when number of symbols read is lowerbound}. 
In (i) and (ii) of Figures~\ref{fig:reading objects when number of symbols read is upperbound} and~\ref{fig:reading objects when number of symbols read is lowerbound} each dataword on left gives one of the four possible cases for Equation~\eqref{eq:TS dataword before simulated computation step}. 
These four cases are given by the four possible values for the pair $(\tne{w_1},z)$ which also determine the four cases in Equation~\eqref{eq:values for w'}. 
Note that if $\TSTC$ behaves as described in the (i) and (ii) of Figures~\ref{fig:reading objects when number of symbols read is upperbound} and~\ref{fig:reading objects when number of symbols read is lowerbound}, then the correct value for $w'$ is appended for each of the four cases in Equation~\eqref{eq:values for w'}.
So proving the correctness of (i) and (ii) in Figures~\ref{fig:reading objects when number of symbols read is upperbound} and~\ref{fig:reading objects when number of symbols read is lowerbound} verifies that the correct value for $w'$ is appended. 

From (i) and (ii) of Figures~\ref{fig:reading objects when number of symbols read is upperbound} and~\ref{fig:reading objects when number of symbols read is lowerbound}, when $\tne{w_1}\in\{\tne{1},\tne{0}\}$ is entered with shift $z=(z_1m+z_2d)\bmod\beta$, the object immediately to the right is entered with a shift of the form $(z_1(m+1)+z_2d)\bmod\beta$. 
So from Equation~\eqref{eq:TS dataword before simulated computation step}, we enter the leftmost object in the word $\{\encodeDeletion,\encodeDeletionPrime\}^\ast\tne{w_2}$ with shift $(z_1(m+1)+z_2d)\bmod\beta$. From Table~\ref{tab:equalities} an $\encodeDeletion$ has a shift change of 0 and $\encodeDeletionPrime$ has a shift change of $z_2$, and so from Lemma~\ref{lem:shiftChange} each object in the word $\{\encodeDeletion,\encodeDeletionPrime\}^\ast\tne{w_2}$ is entered with a shift of the form $z'=(z_1(m+1)+z_2d')\bmod\beta$. Here $0\leqslant d'<3x+1$ since $0=z_2(3x+1)\bmod\beta$. 
From (iii) and (iv) in Figure~\ref{fig:reading objects when number of symbols read is upperbound} and (iii) in Figure~\ref{fig:reading objects when number of symbols read is lowerbound}, when $\encodeDeletion$ and $\encodeDeletionPrime$ objects are entered with a shift of the form $z'=(z_1(m+1)+z_2d')\bmod\beta$, they append only $\encodeDeletion$ and $\encodeDeletionPrime$ objects. 
(From Section~\ref{sec:length of objects}, the range of values for $m$ and $d$ cover all possible shift values, and thus shift values of the from $(z_1(m+1)+z_2d')\bmod\beta$ are covered by the cases in Figures~\ref{fig:reading objects when number of symbols read is upperbound} and~\ref{fig:reading objects when number of symbols read is lowerbound}). 
So proving the correctness of (iii) and (iv) in Figure~\ref{fig:reading objects when number of symbols read is upperbound} and (iii) in Figure~\ref{fig:reading objects when number of symbols read is lowerbound} verifies that reading the word $\{\encodeDeletion,\encodeDeletionPrime\}^\ast$ in Equation~\eqref{eq:TS dataword before simulated computation step} appends a word of the from $\{\encodeDeletion,\encodeDeletionPrime\}^\ast$. 

From the two paragraphs above it follows that if Figures~\ref{fig:reading objects when number of symbols read is upperbound} and~\ref{fig:reading objects when number of symbols read is lowerbound} are correct, then given a dataword of the form shown in Equation~\eqref{eq:TS dataword before simulated computation step}, $\TSTC$ produces a dataword of the form shown in Equation~\eqref{eq:TS dataword after simulated computation step}. 
We complete the proof of this lemma by demonstrating the correctness of Figures~\ref{fig:reading objects when number of symbols read is upperbound} and~\ref{fig:reading objects when number of symbols read is lowerbound}. 
In each line of Figures~\ref{fig:reading objects when number of symbols read is upperbound} and~\ref{fig:reading objects when number of symbols read is lowerbound} the shift with which object $\tne{a_2}$ is entered follows immediately from the shift change caused by reading the leftmost object (see Lemma~\ref{lem:shiftChange} and Table~\ref{tab:equalities}). 
For example in Figure~\ref{fig:reading objects when number of symbols read is upperbound} (i) an $\encodeOne$ is entered with shift $z$, and because $\encodeOne$ has a shift change of $z_1$, the object $\tne{a_2}$ immediately to the right is entered with shift $z+z_1$.
Given that we can show that $\tne{a_2}$ is entered with the correct shift for each case in Figures~\ref{fig:reading objects when number of symbols read is upperbound} and~\ref{fig:reading objects when number of symbols read is lowerbound}, it only remains to show that when each object is read the correct appendant gets appended. 
The method demonstrated in Section~\ref{sec:Reading objects and $u$ subwords} can be applied to show that the word $u$ is defined via Tables~\ref{tab:Tracks in u for Case 1} to~\ref{tab:Tracks in u for Case 2b} such that when each object is read it appends the appendants as shown in Figures~\ref{fig:reading objects when number of symbols read is upperbound} and~\ref{fig:reading objects when number of symbols read is lowerbound}. 
In Section~\ref{sec:Reading objects and $u$ subwords} the method was applied to only the case in Figure~\ref{fig:reading objects when number of symbols read is lowerbound} (i), however in Section~\ref{sec:Reading objects and $u$ subwords} it was also explained how to apply the method to the remaining cases in Figures~\ref{fig:reading objects when number of symbols read is upperbound} and~\ref{fig:reading objects when number of symbols read is lowerbound} and so we do not give them here. 

In the previous paragraph we have shown how to verify that in each object the tracks for $u$ in Tables~\ref{tab:Tracks in u for Case 2a} to~\ref{tab:Tracks in u for Case 2b} will append the appendants shown in Figures~\ref{fig:reading objects when number of symbols read is upperbound} and~\ref{fig:reading objects when number of symbols read is lowerbound}. 
To complete our proof we show that all possible tracks read in $u$ are given in Tables~\ref{tab:Tracks in u for Case 2a} to~\ref{tab:Tracks in u for Case 2b}, and that one and only one value has been assigned to each track.
Using the method in paragraph 2 of Section~\ref{sec:Reading objects and $u$ subwords} and Figure~\ref{fig:u track and encodedOne track} (ii),
one can to verify that the set of all tracks for an object entered with shift $z$ appears in Tables~\ref{tab:Tracks in u for Case 1} to~\ref{tab:Tracks in u for Case 2b}.
The entire range of values for $z$ (given in  Section~\ref{sec:length of objects}) is covered in Tables~\ref{tab:Tracks in u for Case 1} to~\ref{tab:Tracks in u for Case 2b}.
So all possible $u$ tracks are given in Tables~\ref{tab:Tracks in u for Case 1} to~\ref{tab:Tracks in u for Case 2b} and from Lemma~\ref{lem:Correctness of u} each $u$ track in Tables~\ref{tab:Tracks in u for Case 1} to~\ref{tab:Tracks in u for Case 2b} is assigned one and only one value. 
\end{proof} 

The remainder of this section contains lemmas used in the proof of 
Lemma~\ref{lem:simulating a computation step of C} and in explanations in Section~\ref{sec:Simulation algorithm}. Before each lemma we briefly explain its significance for our algorithm.

Section~\ref{sec:Reading objects and $u$ subwords} shows how to verify that in each object the tracks for $u$ in Tables~\ref{tab:Tracks in u for Case 2a} to~\ref{tab:Tracks in u for Case 2b} will append the appendants shown in Figures~\ref{fig:reading objects when number of symbols read is upperbound} and~\ref{fig:reading objects when number of symbols read is lowerbound}.
Lemma~\ref{lem:Correctness of u} shows that there are no contradictions in Tables~\ref{tab:Tracks in u for Case 1},~\ref{tab:Tracks in u for Case 2a} and~\ref{tab:Tracks in u for Case 2b}.
\begin{lemma}\label{lem:Correctness of u}
Each track of the form $\track{s}{u}$ in Tables~\ref{tab:Tracks in u for Case 1},~\ref{tab:Tracks in u for Case 2a}, and~\ref{tab:Tracks in u for Case 2b} has been assigned one and only one value.  
\end{lemma}
\begin{proof} 
Below we state each case followed by the rows from Tables~\ref{tab:Tracks in u for Case 1},~\ref{tab:Tracks in u for Case 2a}, and~\ref{tab:Tracks in u for Case 2b} to which the case applies and then we give the proof of that case.

Case 1: Tracks of the form $\track{(z+i(3x-2)-10)\bmod \beta}{u}$ for $0\leqslant i< \frac{x}{2}-7$ (rows 14 to 18 in Table~\ref{tab:Tracks in u for Case 1} and rows 1 to 7 in Table~\ref{tab:Tracks in u for Case 2b}).
Tracks of this form are used to append the encoding of $\alpha_m$ (see Tables~\ref{tab:Tracks in u for Case 1} and~\ref{tab:Tracks in u for Case 2b}). 
There are three forms of encoding for $\alpha_m$ (see Equations~\eqref{eq:alpha m} to~\eqref{eq:alphaPrime m2}). 
The encoding that is used depends on the shift $z$ with which $\encodeOne$ is entered. Equation~\eqref{eq:alpha m} is used when $z<\beta-z_1$ and the choice between Equations~\eqref{eq:alphaPrime m1} and~\eqref{eq:alphaPrime m2} depend on whether $j<v$ or $j\geqslant v$ (see third column of Table~\ref{tab:Tracks in u for Case 2b}).
Note from rows 7, 9 and 11 of Table~\ref{tab:Tracks in u for Case 2b} that the value of $j$ depends on the value of $z$.
So the shift value $z$ determines which of the encodings in Equations~\eqref{eq:alpha m} to~\eqref{eq:alphaPrime m2} to choose for $\alpha_m$, and from Lemma~\ref{lem:Encode appendant} the shift value $z=((z_1m+z_2d)\bmod\beta)$ encodes $\alpha_m$ and only $\alpha_m$. 
It follows that there is one and only one encoding associated with each $z$.
From Lemma~\ref{lem: ui words with unique shift values} each track of the form $\track{(z+i(3x-2)-10)\bmod \beta}{u}$ for $0\leqslant i< \frac{x}{2}-7$ is entered if and only if an $\encodeOne$ is entered with shift $z$ and the $i+1^{\textrm{th}}$ $u$ is being read. 
This means that track $\track{(z+i(3x-2)-10)\bmod \beta}{u}$ will be read if and only if we are appending the $i+1^\textrm{th}$ object in the encoding of appendant $\alpha_m$ as described in the third paragraph of Section~\ref{sec:Reading objects and $u$ subwords}.
Since there is one and only one encoding associated with each $z$ it follows that each track of the from $\track{(z+i(3x-2)-10)\bmod \beta}{u}$ for $0\leqslant i< \frac{x}{2}-7$ is assigned one and only one value.

Case 2: Tracks of the form $\track{s}{u}$, where $s<\beta-3x$ and $s\neq (z+i(3x-2)-10)\bmod \beta$ for $0\leqslant i< \frac{x}{2}-7$ (rows 2, 5, 6, 7, 10, 11, 12, 19 and 20 in Table~\ref{tab:Tracks in u for Case 1}, rows 1, 3, 5, 7, 9, and 11 in Table~\ref{tab:Tracks in u for Case 2a}, and rows 8 and 10 in Table~\ref{tab:Tracks in u for Case 2b}). All $u$ tracks for this case append only $\encodeDeletion$ objects. This can be easily verified by checking the above mentioned rows. It follows that each track from this case is assigned one and only one value.

Case 3: Tracks of the form $\track{s}{u}$, where $s\geqslant\beta-3x$ and $s\neq (z+i(3x-2)-10)\bmod \beta$ for $0\leqslant i< \frac{x}{2}-7$ (rows 1, 3, 4, 8, 9, 13 and 21 in Table~\ref{tab:Tracks in u for Case 1}, and rows 2, 4, 6, 8, 10 and 12 in Table~\ref{tab:Tracks in u for Case 2a} and rows 9 and 11 in Table~\ref{tab:Tracks in u for Case 2b}). The $u$ tracks entered with shift $\geqslant\beta-3x$  either append $\encodeDeletionPrime$, or result in one of the special cases given by rows 1, 3, 4, 8, and 21 of Table~\ref{tab:Tracks in u for Case 1}. 
Here we omit rows 9 and 13 as they define the same tracks as rows 4 and 8. 
No two of the cases from rows 1, 3, 4, 8, and 21 in Table~\ref{tab:Tracks in u for Case 1} have the same underscript and thus no conflicts occurs when comparing one special case with another special case. 
Note that to see the difference between the underscripts in these special cases one should keep in mind that $x>14$. For example, to show that the row 8 underscript $\beta-2x-8$ is not equal to the row 21 underscript $\beta-3x+2$ we must have $x>10$.
Given that no conflicts occurs between the special cases, it only remains to show that these special cases have no conflicts with tracks that append $\encodeDeletionPrime$. 
The $u$ tracks that append $\encodeDeletionPrime$ are given by rows 2, 4, 6, 8, 10 and 12 in Table~\ref{tab:Tracks in u for Case 2a} and rows 9 and 11 of Table~\ref{tab:Tracks in u for Case 2b}. 
(We need not consider row 7 of Table~\ref{tab:Tracks in u for Case 2b} as this was covered by Case 1.) 
For these cases comparing the underscripts $\bmod\, x$ shows that the values in the underscripts of some of the special cases differ from those that append $\encodeDeletionPrime$. 
For example, in Table~\ref{tab:Tracks in u for Case 1} the underscript in row 8 gives $x-8=(\beta-2x-8)\bmod x$ (since $0=\beta\bmod x$) and in Table~\ref{tab:Tracks in u for Case 2a} the underscript in row 4 gives $x-6=(z+3xi-6)\bmod x$ (since $0=z\bmod x$). 
It follows that there is no conflict between these two rows as they define tracks in $u$ for two different shift values. 
Comparing the underscript values $\bmod\, x$ does not work for all cases as some underscript values may have different shift values but the same values $\bmod\, x$. 
In each of these cases by looking at the range of values in the third column of the tables one sees that the shift values in the underscripts do not coincide. 
For example, in row 8 of Table~\ref{tab:Tracks in u for Case 1} we have a shift of $\beta-2x-8$ and in row 6 of Table~\ref{tab:Tracks in u for Case 2a} we have a shift $z+3x(x-1)-8$. 
From the third column of row 6 in Table~\ref{tab:Tracks in u for Case 2a}, we have the range of values $\beta-2x< z+3x(x-1)-8<\beta$ and so here $z+3x(x-1)-8\neq \beta-2x-8$.
Using this method one finds for the remaining cases (i.e. those not covered by comparing underscripts $\bmod\, x$) that no underscript for a $u$ track from rows 1, 3, 4, 8, and 21 has the same value as an underscript for a $u$ track given by rows 2, 4, 6, 8, 10 and 12 in Table~\ref{tab:Tracks in u for Case 2a} and rows 9 and 11 of Table~\ref{tab:Tracks in u for Case 2b}.
So tracks of the form $\track{s}{u}$ are assigned one and only one value, where $s\geqslant\beta-3x$ and $s\neq (z+i(3x-2)-10)\bmod \beta$ for $0\leqslant i< \frac{x}{2}-7$.
\end{proof}

Recall from paragraph 3 of Section~\ref{sec:Reading objects and $u$ subwords}, that when an $\encodeOne$ is entered with shift $z=(z_1m+z_2d)\bmod\beta$, the $i+1^{\textrm{th}}$ $u$ read appends the $i+1^{\textrm{th}}$ object from the left in the encoding of $\alpha_m=\sigma_1\sigma_2\ldots\sigma_v$. 
So the shift value for the track read in the $i+1^{\textrm{th}}$ $u$ from the left when $\encodeOne$ is entered with shift $z=(z_1m+z_2d)\bmod\beta$ must be unique.
From Section~\ref{sec:Cyclic tag system C'}, $v\leqslant r<\frac{x}{2}-7$ and so we need not concern ourselves with values where $i\geqslant\frac{x}{2}-7$. 
From Figure~\ref{fig:u track and encodedOne track} (ii), for $0\leqslant i<\frac{x}{2}-7$ track $\track{(z+i(3x-2)-10)\bmod\beta}{u}$ is read in the $i+1^{\textrm{th}}$ $u$ from the left when $\encodeOne$ is entered with shift $z$. 
For this reason, in Lemma~\ref{lem: ui words with unique shift values} we show for $0\leqslant k<\frac{x}{2}-7$ that $\track{(z'+k(3x-2)-10)\bmod\beta}{u}$ is read if and only if $\encodeOne$ is entered with shift $z'$ and we are reading $k+1^{\textrm{th}}$ $u$ from the left.

\begin{lemma}\label{lem: ui words with unique shift values}
Let $z=(z_1m+z_2d)\bmod\beta$ be the shift with which we enter $\encodeDeletion$, $\encodeDeletionPrime$, $\encodeZero$, and $\encodeOne$ objects. Then track $\track{(z'+k(3x-2)-10)\bmod\beta}{u}$ is read if and only if $\encodeOne$ is entered with shift $z'$ and the $k+1^{\textrm{th}}$ $u$ from the left is being read. Here $0\leqslant k<\frac{x}{2}-7$ and $z'=(z_1m'+z_2d')\bmod\beta$.
\end{lemma}
\begin{proof}
From Figure~\ref{fig:u track and encodedOne track} (ii) and paragraph 2 of Section~\ref{sec:Reading objects and $u$ subwords}, we know that if an $\encodeOne$ is entered with shift $z'=(z_1m'+z_2d')\bmod\beta$, then the track read in the $k+1^{\textrm{th}}$ $u$ from the left is $\track{(z'+k(3x-2)-10)\bmod\beta}{u}$, where $0\leqslant k<\frac{x}{2}-7$. To complete the proof we show that for any arbitrary track $\track{s}{u}$ that $s\neq (z'+k(3x-2)-10)\bmod\beta$ when $\track{s}{u}$ is \emph{not} read in the $k+1^{\textrm{th}}$ $u$ from the left in an $\encodeOne$ entered with shift $z'$. 
The values for $s$ when entering the objects $\encodeDeletion$, $\encodeDeletionPrime$, $\encodeZero$, and $\encodeOne$ with shift $z$ are given by the underscripts of the $u$ tracks in the middle column of Tables~\ref{tab:Tracks in u for Case 1} to~\ref{tab:Tracks in u for Case 2b}. 
In these tables the value $s$ in $\track{s}{u}$ is of the form $s=(z+y)\bmod\beta$ (for example $s=(z-2)\bmod\beta$ in rows 1 to 3 of Table~\ref{tab:Tracks in u for Case 1}). 
So to show $s\neq (z'+k(3x-2)-10)\bmod\beta$ we prove $((z+y)\bmod\beta) \neq ((z'+k(3x-2)-10)\bmod\beta) $, which we rewrite as $((z-z')\bmod\beta) \neq ((k(3x-2)-10-y)\bmod\beta)$. Because $0=\beta\bmod x$ and $0=(z-z')\bmod x$ for all  $z,z'\in\{(z_1m+z_2d)\bmod\beta\}$ it is sufficient to show that 
\begin{equation}\label{eq:unique tracks in encoded 1}
0\neq(k(3x-2)-10-y)\bmod x 
\end{equation}
In Tables~\ref{tab:Tracks in u for Case 1} to~\ref{tab:Tracks in u for Case 2b} the value $i$ denotes the position of the $u$ word read in an object. For example, $i=0$ is the leftmost $u$ in the object, $i=1$ is the second $u$ from the left and so on. Below we use the value $i$ to give the cases for $\track{s}{u}$ in each object.

Case 1: Reading track $\track{s}{u}$ in $\encodeOne$ when $0\leqslant i<\frac{x}{2}-7$. From Tables~\ref{tab:Tracks in u for Case 1} and~\ref{tab:Tracks in u for Case 2b} we have $s=((z+i(3x-2)-10)\bmod\beta$, and from the previous paragraph we have $y=i(3x-2)-10$. When we substitute this value for $y$ in Equation~\eqref{eq:unique tracks in encoded 1} we get the inequality $0\neq((k-i)(3x-2))\bmod x$ which holds when $k\neq i$, $0\leqslant k<\frac{x}{2}-7$ and $0\leqslant i<\frac{x}{2}-7$. Note that here we do not consider the case $k=i$ as this implies that $z=z'$ which means all the requirements (as given in the lemma statement) for reading $\track{(z'+k(3x-2)-10)\bmod\beta}{u}$ have been met.

Case 2: Reading track $\track{s}{u}$ in $\encodeOne$ when $\frac{x}{2}-7\leqslant i< x$. From Tables~\ref{tab:Tracks in u for Case 1} and~\ref{tab:Tracks in u for Case 2b} we have $s=(z+3xi-x+4)$, and from paragraph 1 of this lemma we have $y=3xi-x+4$. When we substitute this value for $y$ in Equation~\eqref{eq:unique tracks in encoded 1} we get the inequality $0\neq(k(3x-2)-3xi+x-14)\bmod x$ which holds for $0\leqslant k<\frac{x}{2}-7$.

Case 3: Reading track $\track{s}{u}$ in $\encodeOne$ when $i=x$. From Tables~\ref{tab:Tracks in u for Case 1} and~\ref{tab:Tracks in u for Case 2b} we have $s=(z+3x^{2}-x+2)$, and from paragraph 1 of this lemma we have $y=3x^{2}-x+2$.  When we substitute this value for $y$ in Equation~\eqref{eq:unique tracks in encoded 1} we get the inequality $0\neq(k(3x-2)-3x^2+x-12)\bmod x$ which holds for $0\leqslant k<\frac{x}{2}-7$.

Case 4: Reading track $\track{s}{u}$ in $\encodeZero$ or $\encodeDeletionPrime$ when $i=0$. From Tables~\ref{tab:Tracks in u for Case 1} and~\ref{tab:Tracks in u for Case 2a} we have $s=(z-4)$, and from paragraph 1 of this lemma we have $y=-4$. When we substitute this value for $y$ in Equation~\eqref{eq:unique tracks in encoded 1} we get the inequality $0\neq(k(3x-2)-6)\bmod x$ which holds for $0\leqslant k<\frac{x}{2}-7$.

Case 5: Reading track $\track{s}{u}$ in $\encodeZero$ when $1\leqslant i<x$ or in $\encodeDeletionPrime$ when $1\leqslant i<x-1$. From Tables~\ref{tab:Tracks in u for Case 1} and~\ref{tab:Tracks in u for Case 2a} we have we have $s=(z+3xi-6)$, and from paragraph 1 of this lemma we have $y=3xi-6$. When we substitute this value for $y$ in Equation~\eqref{eq:unique tracks in encoded 1} we get the inequality $0\neq(k(3x-2)-3xi-4)\bmod x$ which holds for $0\leqslant k<\frac{x}{2}-7$.

Case 6: Reading track $\track{s}{u}$ in $\encodeZero$ when $i=x$. From Tables~\ref{tab:Tracks in u for Case 1} and~\ref{tab:Tracks in u for Case 2a} we have we have $s=(z+3x^2-8)$, and from paragraph 1 of this lemma we have $y=3x^2-8$. When we substitute this value for $y$ in Equation~\eqref{eq:unique tracks in encoded 1} we get the inequality $0\neq(k(3x-2)-3x^2-2)\bmod x$ which holds for $0\leqslant k<\frac{x}{2}-7$.

Case 7: Reading track $\track{s}{u}$ in $\encodeDeletionPrime$ when $i=x-1$. From Tables~\ref{tab:Tracks in u for Case 1} and~\ref{tab:Tracks in u for Case 2a} we have we have $s=(z+3x(x-1)-8)$, and from paragraph 1 of this lemma we have $y=3x(x-1)-8$. When we substitute this value for $y$ in Equation~\eqref{eq:unique tracks in encoded 1} we get the inequality $0\neq(k(3x-2)-3x(x-1)-2)\bmod x$ which holds for $0\leqslant k<\frac{x}{2}-7$.

Case 8: Reading track $\track{s}{u}$ in $\encodeDeletion$ when $i=0$. From Table~\ref{tab:Tracks in u for Case 1} we have we have $s=(z-2)$, and from paragraph 1 of this lemma we have $y=-2$. When we substitute this value for $y$ in Equation~\eqref{eq:unique tracks in encoded 1} we get the inequality $0\neq(k(3x-2)-8)\bmod x$ which holds for $0\leqslant k<\frac{x}{2}-7$.
\end{proof}

The following lemma shows that each shift $z=(z_1m+z_2d)\bmod\beta$ encodes one and only one appendant $\alpha_{m}$.
The variables in the Lemma statement are from Table~\ref{tab:equalities}.
\begin{lemma}\label{lem:Encode appendant}
For each pair $z=((z_1m+z_2d)\bmod\beta)$ and $z'=((z_1m'+z_2d')\bmod\beta)$, if $m\neq m'$ then $z\neq z'$.
\end{lemma}
\begin{proof}
From the values in Table~\ref{tab:equalities}, we get $z_2(3x+1)=z_1(3x-2)=\beta$. 
Note that $0\leqslant m< 3x-2$ and $0\leqslant d< 3x+1$,  and so we have $z_1m+z_2d<z_1(3x-2)+z_2(3x+1)=2\beta$ (and similarly  $z_1m'+z_2d'<2\beta$). 
So if $z=z'$, then either $z_1m+z_2d=z_1m'+z_2d'$ or $z_1m+z_2d=z_1m'+z_2d'-\beta$ (here we can assume $z<z'$ as the argument is the same for $z'<z$).
We rewrite these case as $z_2(d-d')=z_1(m'-m)$ and $z_2(d-d')=z_1(m'-m-3x+2)$. 
Note that $z_1(m'-m-3x+2)\neq 0$, and since $m\neq m'$ we also have $z_1(m'-m)\neq 0$, which means that for both cases $d\neq d'$. 
From the values in Table~\ref{tab:equalities} we have $z_1=x(3x+1)$, and since $z_2$ and $3x+1$ are relatively prime we get $0\neq (z_2(d-d')\bmod (3x+1))$ for $d\neq d'$, $0\leqslant d< 3x+1$ and $0\leqslant d'< 3x+1$. It follows that the above equalities do not hold since $z_1(m'-m)$ and $z_1(m'-m-3x+2)$ are divisible by $3x+1$ and $z_2(d-d')$ is not, and thus $z\neq z'$.
\end{proof}

\subsection{The halting problem for binary tag systems}

\begin{corollary}\label{cor:2 symbol tag halting problem}
The halting problem for binary tag systems is undecidable.
\end{corollary}
\begin{proof}
In Theorem~\ref{thm:main theorem} the $u$ tracks at odd valued shifts are never read by $\TSTC$. So setting $u$ tracks at odd shifts to be sequences of all $b$ symbols causes no change in the simulation algorithm. $\TSTC$ simulates the cyclic tag system in~\cite{NearyWoods2006C} which has a special appendant $\alpha_h$ that is appended if and only if the Turing machine it simulates is halting. 
We can alter $\TSTC$ so that instead of appending $\alpha_h$ when $C$ halts, it appends an object of odd length so that all subsequent $u$ subwords are entered with an odd shift. This means that a single round on the tag system dataword changes everything to $b$ symbols. 
Now the rule $b\rightarrow b$, which appends one $b$ and deletes $\beta$ symbols, is repeated until number of symbols is $<\beta$ and the computation halts. So the computation halts if and only if the cyclic tag system is simulating a halting Turing machine.
\end{proof}

\section{The Post correspondence problem for 4 pairs of words}
In Theorem~\ref{thm:PCP 4 is undecidable} we show that the Post correspondence problem is undecidable for 4 pairs of words. Theorem~\ref{thm:PCP 4 is undecidable} is proved by reducing the halting problem for the binary tag system given in Lemma~\ref{lem:binary tag system for PCP} to the Post correspondence problem. The halting problem for the binary tag system in Lemma~\ref{lem:binary tag system for PCP} is proved undecidability by simulating the cyclic tag system given in Lemma~\ref{lem:CTS with input dataword 1}. 

\begin{definition}[Post correspondence problem]\label{def:PCP}
Given a set of pairs of words $\{(r_i,v_i)| r_i,v_i\in\Sigma^\ast, 0\leqslant i\leqslant n\}$ where $\Sigma$ is a finite alphabet, determine whether or not there is a non-empty sequence $r_{i_1}r_{i_2}\ldots r_{i_l}=v_{i_1}v_{i_2}\ldots v_{i_l}$. 
\end{definition}

\begin{lemma}\label{lem:CTS with input dataword 1}
Let $\CTS=\alpha_0,\alpha_1\ldots,\alpha_{p-1}$ be a cyclic tag system and let $w$ be an input dataword to $\CTS$. Then there is a cyclic tag system $\CTS_w$ that takes a single 1 as its input and simulates the computation of $\CTS$ on $w$.  
\end{lemma}
\begin{proof}
The binary dataword $w=w_1w_2w_3\ldots w_n$ is encoded as $\tne{w}= w_10w_20w_30\ldots w_n0$, and each binary appendant $\alpha_m=\sigma_1\sigma_2\sigma_3\ldots \sigma_m$ in $C$ is encoded as $\tne{\alpha_m}=\sigma_10\sigma_20\sigma_30\ldots \sigma_m 0$. The program for $\CTS_w$ is defined by the equation
\begin{equation*}\label{eq:CTS that takes 1 as input}
\CTS_w=\tne{w},\tne{\alpha_0},\epsilon,\tne{\alpha_1},\epsilon,\tne{\alpha_2},\epsilon,\tne{\alpha_3}\ldots,\epsilon\tne{\alpha_{p-1}}
\end{equation*}
where $\epsilon$ is the empty word and $\tne{\alpha_i}$ is defined above. The configuration for $\CTS_w$ at the start of the computation is given by
\begin{xalignat*}{2}
\pmb{\tne{w}},\tne{\alpha_0},\epsilon,\tne{\alpha_1},\epsilon,\tne{\alpha_2},\ldots,\epsilon\tne{\alpha_{p-1}}& & &1
\end{xalignat*}
where the program is given on the left with the marked appendant $\tne{w}$ in bold, and the input dataword is a single 1 and is given on the right. After the first computation step we have
\begin{xalignat*}{2}
\tne{w},\pmb{\tne{\alpha_0}},\epsilon,\tne{\alpha_1},\epsilon,\tne{\alpha_2},\ldots,\epsilon\tne{\alpha_{p-1}}& & &w_10w_20w_30\ldots w_n0
\end{xalignat*}
In the configuration above the encoding $\tne{w}=w_10w_20w_30\ldots w_n0$ of $w$ has been appended. Now the simulation of the first computation step of $C$ on $w$ begins. Every second appendant in $\CTS_w$ is an $\tne{\alpha_j}$ appendant and every second symbol in the dataword of $\CTS_w$ is a $w_i$ symbol. So $\CTS_w$ on the input dataword 1 simulates the computation of $\CTS$ on $w$.
\end{proof}

\begin{lemma}\label{lem:binary tag system for PCP}
The halting problem is undecidable for binary tag systems with deletion number $\beta$, alphabet $\{b,c\}$ and rules of the form $b\rightarrow b$ and $c\rightarrow u_1\ldots u_lb$ ($u_i\in\{b,c\}$), when given $u_{\beta}u_{\beta+1}\ldots u_lb$ as input.
\end{lemma}
\begin{proof}
We use the tag system $\TSTC$ from Theorem~\ref{thm:main theorem} to construct a 2-symbol tag system $\TSTCp$ of the type mentioned in the lemma statement. From Lemma~\ref{lem:CTS with input dataword 1}, we can assume without loss of generality that $\TSTC$ simulates cyclic tag systems whose input is a single 1.

Recall that $\TSTC$ has rules of the form $b\rightarrow b$ and $c\rightarrow u$ (where $u=u_0\ldots u_l\in\{b,c\}^\ast$), and a deletion number $\beta$. 
In $\TSTC$ track $\track{0}{u}=u_{0}u_{\beta}u_{2\beta}\ldots u_{3x\beta}$ is never read, and so we can define $\track{0}{u}$ such that reading $u_{\beta}u_{2\beta}\ldots u_{3x\beta}$ appends the word $\encodeOne'=b^{10}(ubb)^{\frac{x}{2}-7}u^{\frac{x}{2}+6}b^{2}ub^{x+2}$. 
Note that $\encodeOne'$ is obtained from $\encodeOne=\encodeOneSequence$ by removing a single $u$ form the subword $u^{\frac{x}{2}+7}$. 
Because each $u$ in the subword $u^{\frac{x}{2}+7}$ appends a garbage object that has no effect on the computation, reading an $\encodeOne'$ simulates reading an $\encodeOne$. 
So given the input dataword $u_{\beta}u_{\beta+1}\ldots u_lb$ the sequence $u_{\beta}u_{2\beta}\ldots u_{3x\beta}$ is read appending $\encodeOne'$ and the simulation of $\CTS_w$ on input 1 is ready to begin.

Now we replace the rule $c\rightarrow u$ in $\TSTC$ with the rule $c\rightarrow u'$ (where $u'=u_1\ldots u_lb$)
The new track that we defined above for appending $\encodeOne'$ is $\track{0}{u}$ with its first symbol $u_{0}$ deleted and for this reason we can delete $u_0$ from $u$ to give $u'$ as it is never read.
The extra $b$ added at the right end of $u'$ means that $|u'|=|u|$. 
Following the replacement of $c\rightarrow u$ with $c\rightarrow u'$, we make a few minor changes which we detail below so that the simulation of $\CTS_w$ on input 1 proceeds correctly.

The shift change from reading the input word $u_{\beta}u_{\beta+1}\ldots u_lb$ is $3x-1$ (as $|u_{\beta}u_{\beta+1}\ldots u_lb|=|u|-\beta+1$ and the shift change for $u$ is $3x$). 
Recall that reading $u_{\beta}u_{\beta+1}\ldots u_lb$ appends $\encodeOne'$, so after reading the input $u_{\beta}u_{\beta+1}\ldots u_lb$ we enter $\encodeOne'$ with the shift value $3x-1$.
Above to get $u'$ we deleted the leftmost symbol $u_{0}$ from $u$ and so every track in $u$ is shifted one symbol to the left in $u'$. A shift of $-1$ corrects for this and so the shift change of $3x-1$ is in fact equivalent to a shift change of $3x$.
A shift of $3x$ simulates that the marked appendant is at $\alpha_1$ instead of $\alpha_0$. 
To see this note from Table~\ref{tab:equalities} that $3x=((z_1m+z_2d)\bmod\beta)$ for $d=3x$ and $m=1$, which encodes that $\alpha_{m}=\alpha_{1}$ is the marked appendant. 
If we alter the simulated cyclic tag system by taking the last appendant in the program and placing it at the start of the list of appendants, then every appendant get shifted one place to the right in the circular program. 
Now when we enter $\encodeOne'$ with shift $3x$, we are simulating the marker at the correct encoded appendant. 

There is another problem to overcome. 
The object $\encodeOne'$ has one less $u$ than $\encodeOne$ and so has a different shift change to $\encodeOne$.
From the values in Table~\ref{tab:equalities}, the object $\encodeOne'=b^{10}(ubb)^{\frac{x}{2}-7}u^{\frac{x}{2}+6}b^{2}ub^{x+2}$ has a shift change of $z_2=3x^2-x$, and such a shift change value simulates no change in the marked appended (see the end of Section~\ref{sec:Encoding the marked appendant}). 
Recall that we are simulating $\CTS_w$, and so from Lemma~\ref{lem:CTS with input dataword 1}, when we read the encoded $\encodeOne'$ we append the encoding of the dataword $w$. 
If we change the dataword so that we encode the word $0w$ (instead of $w$), then the extra encoded $0$ is read before we be read the encoding of $w$. 
The shift change caused by the extra encoded $0$ simulates the marker moving to the next appendant so that we enter the encoding of $w$ with the correct shift. 
Now that we have successfully appended the encoding of $w$ the remainder of the computation of $\TSTC$ is simulated step for step. 

Finally, the technique from  Corollary~\ref{cor:2 symbol tag halting problem} can be use to modify the above system so that it halts if an only if it is simulating a halting Turing machine. Note that because the tracks from $u$ are shifted one place to the left in $u'$ when we apply the technique from  Corollary~\ref{cor:2 symbol tag halting problem} we set the even tracks (instead of the odd tracks) in $u'$ to be sequences of all $b$ symbols. This completes our construction of $\TSTCp$.
\end{proof}

\begin{theorem}\label{thm:PCP 4 is undecidable}
The Post correspondence problem is undecidable for 4 pairs of words. 
\end{theorem}
\begin{proof}
We reduce the halting problem for the binary tag system $\TSTCp$ in Lemma~\ref{lem:binary tag system for PCP} to the Post correspondence problem for 4 pairs of words. 
The symbols $b$ and $c$ in $\TSTCp$ are encoded as $\tne{b}=10^{\beta}1$ and $\tne{c}=1$ respectively, where $\beta$ is the deletion number of $\TSTCp$. 
The halting problem for $\TSTCp$ reduces to the Post correspondence problem given by the 4 pairs of binary words
\begin{equation*}\label{eq:undecidable PCP}
\PCPfour=\{(1,\,1\tne{u_1}\tne{u_2}\dots \tne{u_l}10),\;(10^{\beta}1,\,110),\;(10^{\beta},\epsilon),\;(1,0)\}
\end{equation*}
where $\epsilon$ is the empty word and $u_i\in\{0,1\}$. 
Let $w=w_{i_1}w_{i_2}\ldots w_{i_l}$ and $v=v_{i_1}v_{i_2}\ldots v_{i_l}$, where each $(w_{i},v_{i})\in\PCPfour$ and $w$ is a prefix of $v$. 
We will call the pair $(w,v)$ a configuration of $\PCPfour$. Because $\TSTCp$ has an initial input word that ends in a $b$ and both of the rules of $\TSTCp$ append words that end in a $b$ an arbitrary dataword of $\TSTCp$ has the from $x_0x_1\ldots x_{r}b\in\{b,c\}^\ast b$. The arbitrary dataword $x_0x_1\ldots x_{r}b$ is encoded by a $\PCPfour$ configuration of the form
\begin{equation}\label{eq:encoding of T'_C config for P}
(w,v)=(w,\,w\tne{x_0}\tne{x_1}\ldots\tne{x_{r}}10^{\beta})
\end{equation}
In each configuration $(w,v)$, the unmatched part of $v$ (given by $\tne{x_0}\tne{x_1}\ldots\tne{x_{r}}10^{\beta}$) encodes the current dataword of $\TSTCp$.

We must have $(1,\:1\tne{u_1}\tne{u_2}\dots \tne{u_l}10)$ as the leftmost pair $(w_{i_1},v_{i_1})$ in a match as having any other pair from $\PCPfour$ as the leftmost pair will not give a match. 
Starting from the pair $(1,\:1\tne{u_1}\tne{u_2}\dots \tne{u_l}10)$, if $u_1=c$ we add the pair $(1,0)$ and this matches $\tne{c}=1$ simulating the deletion of $u_1$. 
If, on the other hand, $u_1=b$ we add the pair $(10^{\beta},\epsilon)$ followed by the pair $(1,0)$ and this matches $\tne{b}=10^{\beta}1$ simulating the deletion of $u_1$. 
So after matching $\tne{u_1}$ we have $(1\tne{u_1},\:1\tne{u_1}\tne{u_2}\dots \tne{u_l}100)$.
We match $\beta-1$ encoded $\TSTCp$ symbols in this way to give $(w,v)=(1\tne{u_1}\ldots\tne{u_{\beta-1}},\;1\tne{u_1}\tne{u_2}\dots \tne{u_l}10^{\beta})$. 
The configuration is now of the form given in Equation~\eqref{eq:encoding of T'_C config for P} and the unmatched sequence $\tne{u_{\beta}}\dots \tne{u_l}10^{\beta}$ in $v$ encodes the input dataword to $\TSTCp$ in Lemma~\ref{lem:binary tag system for PCP}.

A computation step of $\TSTCp$ on the arbitrary dataword $x_0x_1\ldots x_{r}b$ is of one the two forms:
\begin{align}
 cx_1\ldots x_{r}b&\;\;\;\vdash\;\;\; x_{\beta-1}\ldots x_{r}bu_1\ldots u_lb \label{eq:T'C timestep read symbol c}\\
 bx_1\ldots x_{r}b&\;\;\;\vdash\;\;\; x_{\beta-1}\ldots x_{r}bb \label{eq:T'C timestep read symbol b}
\end{align}

The two forms of computation step given in Equations~\eqref{eq:T'C timestep read symbol c} and~\eqref{eq:T'C timestep read symbol b} are simulated as follows: In Equation~\eqref{eq:encoding of T'_C config for P}, if $x_0=c$ then $\tne{x_0}=1$ and we add the pair $(1,\,1\tne{u_1}\tne{u_2}\dots \tne{u_l}10)$ to simulate the $\TSTCp$ rule $c\rightarrow u_1\ldots u_lb$,  and this gives $(w1,\;w1\tne{x_1}\ldots\tne{x_{r}}10^{\beta}1\tne{u_1}\tne{u_2}\dots \tne{u_l}10)$.
In Equation~\eqref{eq:encoding of T'_C config for P}, if $x_0=b$ then $\tne{x_0}=10^{\beta}1$ and we add the pair $(10^{\beta}1,110)$ to simulate the $\TSTCp$ rule $b\rightarrow b$, and this gives $(w10^{\beta}1,\;w10^{\beta}1\tne{x_1}\ldots\tne{x_{r}}10^{\beta}110)$.
In both cases ($x_0=c$ and $x_0=b$) to complete the simulation of the computation step we continue to match the pairs $(10^{\beta},\epsilon)$ and $(1,0)$ as we did in the previous paragraph to simulate the deletion of a further $\beta-1$ tag system symbols. 
Simulating the deletion of $\beta-1$ symbols adds a further $\beta-1$ of the 0 symbols at the right end of the encoded dataword. 
So if $x_0=c$ this gives $(w1\tne{x_1}\ldots\tne{x_{\beta-1}},\;w1\tne{x_1}\ldots\tne{x_{r}}10^{\beta}1\tne{u_1}\tne{u_2}\dots \tne{u_l}10^{\beta})$, with the unmatched part in this pair encoding the dataword on the right of Equation~\eqref{eq:T'C timestep read symbol c} after the computation step. 
Alternatively, if $x_0=b$ we get $(w10^{\beta}1\tne{x_1}\ldots\tne{x_{\beta-1}},\; w10^{\beta}1\tne{x_1}\ldots\tne{x_{r}}10^{\beta}110^{\beta})$, with the unmatched part in this pair encoding the dataword on the right of Equation~\eqref{eq:T'C timestep read symbol b} after the computation step.
The simulated computation step is now complete.

We now explain how $\PCPfour$ simulates $\TSTCp$ halting with a matching sequence. In $\TSTCp$ the rule $b\rightarrow b$ deletes $\beta$ symbols and append a single $b$ reducing the number of symbols in the dataword by $\beta-1$, and the rule $c\rightarrow u_1\ldots u_l b$ deletes $\beta$ symbols and appends $(3x+1)\beta-3x$ symbols ($|u_1\ldots u_l b|=|u'|=|u|$, see Table~\ref{tab:equalities} and Lemma~\ref{lem:binary tag system for PCP}) increasing the number of symbols in the dataword by $3x(\beta-1)$.
So, because the input dataword $u_{\beta}\dots u_lb$ is of length $3x(\beta-1)+1$ and the rules either increase the length by $3x(\beta-1)$ or decrease it by $\beta-1$, all datawords of $\TSTCp$ have lengths of $y(\beta-1)+1$, where $y\in\Nset$. 
From Corollary~\ref{cor:2 symbol tag halting problem}, $\TSTCp$ halts when the length of its final dataword (which consists entirely of $b$ symbols) is less than the deletion number $\beta$. 
So, when $\TSTCp$ halts we have $y(\beta-1)+1<\beta$ which means the dataword is a single $b$. 
From Equation~\eqref{eq:encoding of T'_C config for P}, this is encoded as the configuration $(w,v)=(w,w10^{\beta})$. 
By appending the pair $(10^{\beta},\epsilon)$ to $(w,v)$, we get the pair of matching sequences $(w10^{\beta},w10^{\beta})$ when $\TSTCp$ halts. 
Note that whenever there is choice of which pair to append, only the choice that follows the simulation as described above has the possibility to lead to a match (all other choices lead to a mismatch). 
Therefore, $\PCPfour$ has a matching sequence if and \mbox{only if $\TSTCp$ halts.}
\end{proof}

\subsection{Undecidability in simple matrix semi-groups}
Undecidability bounds for the Post correspondence problem have been used by a number of authors~\cite{Bell2008,Blondel1997,Bournez2002,Cassaigne1998,Halava2001,Halava2007,Halava2007A,Paterson1970} in the search for undecidable decision problems in simple matrix semi-groups. The undecidability of the Post correspondence problem for 7 pairs of words~\cite{Matiyasevich2005}
has been frequently used to find undecidability in simple matrix semi-groups. Theorem~\ref{thm:PCP 4 is undecidable} results in an immediate improvement on many of these results. Here we will just describe improvements for two of these problems. Of the decision problems on simple matrix semi-groups the mortality problem has in particular received much attention. 

\begin{definition}[Matrix mortality problem]
Given a finite set of $d\times d$ integer matrices $\{M_1,M_2,\ldots$ $M_{n-1},M_n\}$, is there a product $M_{i_1}M_{i_2}\ldots M_{i_k}$the produces the zero matrix, where $1\leqslant i\leqslant n$? 
\end{definition}

To date the best known bounds for the undecidability of the matrix mortality problem are due to Halava et al.~\cite{Halava2007}. Improving on the reduction of Paterson~\cite{Paterson1970}, they showed that the matrix mortality problem is undecidable for sets with seven $3\times 3$ matrices. 
Cassaigne and Karhum\"{a}ki~\cite{Cassaigne1998} showed that if the mortality problem is undecidable for a set of $n$ matrices of dimension $d\times d$ then the mortality problem is undecidable for a pair of $nd\times nd$ matrices. 
So an immediate corollary of the result given by Halava et al.~is that the mortality problem is undecidable a set of two $ 21\times 21$ matrices.
By applying the reductions in~\cite{Halava2001} and~\cite{Cassaigne1998} to $\PCPfour$ in Theorem~\ref{thm:PCP 4 is undecidable} we get Corollary~\ref{cor:matrix mortality}.
\begin{corollary}\label{cor:matrix mortality}
The matrix mortality problem is undecidable for sets with five $3\times 3$ matrices and for sets with two $ 15\times 15$ matrices. 
\end{corollary}

Halava and Hirvensalo~\cite{Halava2007A} give undecidability results for sets that consist of a pair of matrices with remarkably small dimensions. 
One of the problems they tackle is the scalar reachability problem which they prove undecidable for a pair of $9\times 9$ matrices.
\begin{definition}[Scalar reachability problem]
Given a finite set of $d\times d$ integer matrices $\{M_1,M_2,\ldots$ $M_{ n-1},M_n\}$, a vector $y\in\Zset^d$, $x^T$ the transpose of vector $x\in\Zset^d$, and a constant $e\in\Zset$, is there a product $x^TM_{i_1}M_{i_2}\ldots M_{i_k}y=e$? 
\end{definition}
By applying the reductions in~\cite{Halava2007A} to $\PCPfour$ in Theorem~\ref{thm:PCP 4 is undecidable} we get Corollary~\ref{cor:scalar reachability}.
\begin{corollary}\label{cor:scalar reachability}
The scalar reachability problem is undecidable for two $7\times 7$ matrices. 
\end{corollary}

\subsubsection*{Acknowledgements:}
This work was supported by Science Foundation Ireland, grant number 09/RFP/CMS2212 and by Swiss National Science Foundation grant number 200021-141029. I would like to thank Matthew Cook and Damien Woods for their comments and discussions, and Vesa Halava and Mika Hirvensalo for their advice on  undecidability in simple matrix semi-groups.

\bibliographystyle{abbrv}
\bibliography{Binary_tag_systems.bib}

\end{document}